\documentclass{article}%
\usepackage{amssymb}
\usepackage{amsfonts}
\usepackage{amsmath}
\usepackage[numbers,sort&compress]{natbib}
\usepackage{graphicx}%
\providecommand{\U}[1]{\protect\rule{.1in}{.1in}}
\newtheorem{theorem}{Theorem}
\newtheorem{acknowledgement}[theorem]{Acknowledgement}

\newtheorem{definition}[theorem]{Definition}
\newtheorem{example}{Example}

\newtheorem{lemma}{Lemma}

\newenvironment{proof}[1][Proof]{\noindent\textbf{#1.} }{\ \rule{0.5em}{0.5em}}
\numberwithin{equation}{section}
\begin{document}

\title{Systematic construction of non-autonomous Hamiltonian equations of
Painlev\'{e}-type. III. Quantization }
\author{Maciej B\l aszak\\Faculty of Physics, Department of of Mathematical Physics and Computer Modelling,\\A. Mickiewicz University, 61-614 Pozna\'{n}, Poland\\\texttt{blaszakm@amu.edu.pl}
\and Krzysztof Marciniak\\Department of Science and Technology \\Campus Norrk\"{o}ping, Link\"{o}ping University\\601-74 Norrk\"{o}ping, Sweden\\\texttt{krzma@itn.liu.se}}
\maketitle

\begin{abstract}
This is the third article in our series of articles exploring connections
between dynamical systems of St\"{a}ckel-type and of Painlev\'{e}-type. In
this article we present a method of deforming of minimally quantized
quasi-St\"{a}ckel Hamiltonians, considered in Part I to self-adjoint operators
satisfying the quantum Frobenius condition, thus guaranteeing that the
corresponding Schr\"{o}dinger equations posses common, multi-time solutions.
As in the classical case, we obtain here both magnetic and non-magnetic
families of systems. We also show the existence of multitime-dependent quantum
canonical maps between both classes of quantum systems.

\end{abstract}

\section{Introduction}

This is the third article in the suit of articles investigating relations
between Painlev\'{e}-type systems and St\"{a}ckel-type systems. In the first
paper (Part I, i.e. \cite{part1}, see also \cite{arxiv}) we have constructed,
starting from appropriate St\"{a}ckel-type systems, multi-parameter families
of Frobenius integrable non-autonomous Hamiltonian systems with arbitrary
number of degrees of freedom. In the second article (Part II, i.e.
\cite{part2}) we have constructed the isomonodromic Lax representations for
these systems, proving that Frobenius integrable systems constructed in Part I
are indeed Painlev\'{e}-type systems. Each of such families was written in two
different representations, called an ordinary one and a magnetic one,
respectively, and they were connected by a multi-time canonical transformation
\cite{Iwasaki}. In this paper we discuss the minimal quantization
\cite{blasz2016},\cite{Book} of Painlev\'{e}-type systems considered in Part I
and Part II. We prove that this quantization turns the Hamiltonians $H_{r}$ of
the Painlev\'{e}-type systems into their quantum counterparts, that is the
self-adjoint operators $\widehat{H}_{r}$, acting in an appropriate Hilbert
space $\mathcal{H}$ and satisfying \emph{the quantum Frobenius condition}%
\begin{equation}
i\hslash\frac{\partial\widehat{H}_{r}}{\partial t_{s}}-i\hslash\frac
{\partial\widehat{H}_{s}}{\partial t_{r}}+\left[  \widehat{H}_{r}%
,\widehat{H}_{s}\right]  =0,\quad r,s=1,\ldots,n\text{. } \label{jeden}%
\end{equation}
We also prove that both types of quantum Painlev\'{e}-type systems are related
by a quantum canonical transformation which resembles the classical result.

The paper is organized as follows. Section \ref{s2} is devoted to a concise
presentation (with references to literature) of both classical and quantum
St\"{a}ckel systems, both with ordinary and with magnetic potentials, as well
as canonical transformations between them. In Section \ref{sec 2} we remind
the construction of classical Painlev\'{e}-type systems from appropriate
deformations of St\"{a}ckel-type systems. Section \ref{s4} contains the first
of two main results of this paper, namely the systematic method of minimal
quantization of all the classical systems from Section \ref{sec 2} in such a
way that they satisfy the quantum Frobenius condition (\ref{jeden}). In
Section \ref{s5} we present the second result of this article that is the
multi-time quantum canonical transformations between quantum Painlev\'{e}-type
systems with magnetic potentials and the corresponding quantum
Painlev\'{e}-type systems with ordinary potentials.

\section{Classical and quantum St\"{a}ckel systems with ordinary and magnetic
potentials\label{s2}}

Consider the following algebraic curve in the $(x,y)$-plane
\begin{equation}
\sum\limits_{\alpha\in I_{\alpha}}c_{\alpha}x^{\alpha}+\sum\limits_{\gamma\in
I_{\gamma}}d_{\gamma}x^{\gamma}y+\sum_{r=1}^{n}h_{r}x^{n-r}=\frac{1}{2}%
x^{m}y^{2} \label{2.1new}%
\end{equation}
where $m\in\left\{  0,\ldots,n+1\right\}  $, $I_{\alpha}$ and $I_{\gamma}$ are
finite subsets of $\mathbb{Z}$ and where $c_{\alpha}$ and $d_{\gamma}$ are
real constants. Taking $n$ copies of (\ref{2.1new}) at points $(x,y)=(\lambda
_{i},\mu_{i})$, $i=1,\dotsc,n$, we obtain a system of $n$ linear equations
(separation relations) for $h_{r}.$ Solving this system yields $n$ functions
(Hamiltonians)%
\begin{equation}
h_{r}=\frac{1}{2}\mu^{T}A_{r}\mu+\sum\limits_{\gamma\in I_{\gamma}}d_{\gamma
}\mu^{T}P_{r}^{(\gamma)}+\sum\limits_{\alpha\in I_{\alpha}}c_{\alpha}%
V_{r}^{(\alpha)},\quad r=1,\dotsc,n \label{2.2new}%
\end{equation}
on a $2n$-dimensional manifold (phase space) $\mathcal{M}$ parametrized by the
coordinates $(\lambda,\mu)=(\lambda_{1},\ldots,\lambda_{n},\mu_{1},\ldots
\mu_{n})$, where
\[
E_{r}=\frac{1}{2}\mu^{T}A_{r}\mu\equiv\frac{1}{2}\mu^{T}K_{r}G\mu\text{,
\ \ }r=1,\ldots,n,\text{\ \ }%
\]
$G\,$ is a contravariant metric tensor on the configurational space $Q$ (such
that $\mathcal{M}=T^{\ast}Q$), $K_{r}$ ($K_{1}=\operatorname{Id}$) are
$(1,1)$-Killing tensors of $G$ (for any $m$ even though $G$ depends on $m$)
\cite{blasz2007}, $P_{r}^{(\gamma)}$ are basic separable vector potentials and
$V_{r}^{(\alpha)}$ are basic separable scalar potentials (see Part I). By
construction, all the Hamiltonian functions $h_{r}$ are in involution
\begin{equation}
\{h_{r},h_{s}\}\equiv\pi(dh_{r},dh_{s})=0,\ \ \ \ \ \ \ r,s=1,...,n
\label{konwencja}%
\end{equation}
with respect to the Poisson bracket $\pi=%
{\textstyle\sum_{i=1}^{n}}
\frac{\partial}{\partial\lambda_{i}}\wedge\frac{\partial}{\partial\mu_{i}}$ on
$\mathcal{M}$. By construction, they also separate in coordinates
$(\lambda,\mu)$. The Hamiltonians (\ref{2.2new}) are known in literature as
(classical) St\"{a}ckel Hamiltonians, while $E_{r}$ are their geodesic parts.

In the separation coordinates $\lambda_{i},\ i=1,...,n$, the geometric objects
$A_{r}$, $G,$ $K_{r},$ $P_{r}^{(\gamma)}$ and $V_{r}^{(\alpha)}$ are
explicitly given by
\begin{align}
\left(  A_{r}\right)  ^{ij}  &  =-\frac{\partial\rho_{r}}{\partial\lambda_{i}%
}\frac{\lambda_{i}^{m}}{\Delta_{i}}\delta^{ij},\text{ \ \ }G^{ij}%
=\frac{\lambda_{i}^{m}}{\Delta_{i}}\delta^{ij},\ \ \ (K_{r})_{j}^{i}%
=-\frac{\partial\rho_{r}}{\partial\lambda_{i}}\delta_{j}^{i},\nonumber\\
& \label{2.4}\\
\ \ \ (P_{r}^{(\gamma)})^{j}  &  =\frac{\partial\rho_{r}}{\partial\lambda_{j}%
}\frac{\lambda_{j}^{\gamma}}{\Delta_{j}}=-(K_{r})_{j}^{j}G^{jj}\lambda
_{j}^{\gamma-m}=-A_{r}^{jj}\lambda_{j}^{\gamma-m},\ \ \ V_{r}^{(\alpha)}%
=\sum_{j=1}^{n}\frac{\partial\rho_{r}}{\partial\lambda_{j}}\frac{\lambda
_{j}^{\alpha}}{\Delta_{j}}\nonumber
\end{align}
(no summation in the above formulas unless explicitly stated) where
$\Delta_{j}=%
{\textstyle\prod\nolimits_{k\neq j}}
(\lambda_{j}-\lambda_{k})$ and $\rho_{r}\left(  \lambda\right)  =(-1)^{r}%
s_{r}(\lambda)$ where $s_{r}(\lambda)$ are elementary symmetric polynomials.
Note that \thinspace the coordinates $\lambda$ are thus orthogonal coordinates
for the metric $G$.

In what follows we will also work in the so called canonical Vi\`{e}te
coordinates $(q,p)$ on $\mathcal{M}$, related with the separation coordinates
$(\lambda,\mu)$ through the point transformation
\begin{equation}
q_{i}=\rho_{i}(\lambda),\ \ \ \ \text{\thinspace}p_{i}=-\sum_{k=1}^{n}%
\frac{\lambda_{k}^{n-i}}{\Delta_{k}}\mu_{k},\text{ \ \ }i=1,\ldots,n.
\label{Viete}%
\end{equation}
Due to (\ref{2.4}) the metric tensor $G$ for arbitrary $m$ is constructed by
\begin{equation}
G=L^{m}G_{0}, \label{n1}%
\end{equation}
where $L$ is the so called special conformal Killing tensor on $Q$
\cite{Ben1997} and where $G_{0}$ is the metric tensor for $m=0$ (thus
$G_{0}^{ij}=\frac{1}{\Delta_{i}}\delta^{ij}$ and $L_{j}^{i}=\lambda_{i}%
\delta^{ij}$). In Vi\`{e}te coordinates $L$ and $G_{0}$ have the form $\ $
\begin{equation}
L=\left(
\begin{array}
[c]{cccc}%
-q_{1} & 1 & 0 & 0\\
\vdots & 0 & \ddots & 0\\
\vdots & 0 & 0 & 1\\
-q_{n} & 0 & 0 & 0
\end{array}
\right)  ,\ \ \ G_{0}=\left(
\begin{array}
[c]{rrrrr}%
0 & 0 & 0 & 0 & 1\\
0 & 0 & 0 & 1 & q_{1}\\
\vdots & \vdots & \vdots & \vdots & \vdots\\
0 & 1 & q_{1} & \cdots & q_{n-2}\\
1 & q_{1} & \cdots & q_{n-2} & q_{n-1}%
\end{array}
\right)  . \label{n2}%
\end{equation}
Further, the $(1,1)$-Killing tensors $K_{r}$ on $Q$ for every $G$ given by
(\ref{n1}), can be effectively computed through \cite{Ben1997}
\[
K_{1}=\operatorname{Id}\text{, \ \ \ }K_{r}=\sum_{k=0}^{r-1}q_{k}%
L^{r-1-k}\text{, \ \ }r=2,\ldots,n.
\]
The scalar potentials $V_{r}^{(\alpha)}$ can be explicitly constructed by a
recursion formula \cite{blasz2011},\cite{part1}. Finally, the basic separable
vector potentials $P_{r}^{(\gamma)}$ in Vi\`{e}te coordinates have the form
\begin{equation}
\left(  P_{r}^{(\gamma)}\right)  ^{j}=-\sum_{s=0}^{r-1}q_{s}V_{j}%
^{(r+\gamma-s-1)}. \label{Pwq}%
\end{equation}
More information about the structure of the above geometric objects is given
in Part I and Part II.

Let us now turn to the issue of quantization of the classical St\"{a}ckel
Hamiltonians (\ref{2.2new}). A natural way to quantize these Hamiltonians is
by the procedure of \emph{minimal quantization} \cite{blasz2016},\cite{Book},
that we shortly remind here.

\begin{definition}
$\label{minq}$Given a Poisson manifold $\mathcal{M}=T^{\ast}Q$ and a metric
$g$ on $Q$, the minimal quantization of the quadratic in momenta function
$H=\frac{1}{2}C^{kj}p_{k}p_{j}\,$\ on $\mathcal{M}$, where $C^{kj}$ is a
symmetric $(2,0)$-tensor on $Q$, is the second-order linear self-adjoint
operator
\begin{equation}
\widehat{H}=-\frac{1}{2}\hslash^{2}\nabla_{k}C^{kj}\nabla_{j}=-\frac{1}%
{2}\hslash^{2}\left\vert g\right\vert ^{-\frac{1}{2}}\partial_{k}\left\vert
g\right\vert ^{\frac{1}{2}}C^{kj}\partial_{j} \label{qq2}%
\end{equation}
acting in the Hilbert space $\mathcal{H}$ $=L^{2}(Q,\left\vert g\right\vert
^{1/2}dq)$, where $\nabla_{j}$ are covariant derivatives for the Levi-Civita
connection of $g$ while $q$ are (any) variables on $Q$ and $\left\vert
g\right\vert =\det g$. Likewise, the minimal quantization of the linear in
momenta function $W=X^{j}p_{j}$, where $X^{j}$ are components of the vector
field $X=X^{j}\frac{\partial}{\partial q_{j}}$ on $Q$, is the first-order
linear operator
\begin{equation}
\widehat{W}=-\frac{1}{2}i\hslash\left(  \nabla_{j}X^{j}+X^{j}\nabla
_{j}\right)  =-\frac{1}{2}i\hslash\left(  \left\vert g\right\vert ^{-\frac
{1}{2}}\partial_{j}\left\vert g\right\vert ^{\frac{1}{2}}X^{j}+X^{j}%
\partial_{j}\right)  , \label{qq3}%
\end{equation}
acting in $\mathcal{H}$ $=L^{2}(Q,\left\vert g\right\vert ^{1/2}dq)$. Finally,
the minimal quantization of a function $f$ on $Q$ is the operator
$\widehat{f}$ of pointwise multiplication by $f$ so that $\widehat{f}=f$.
\end{definition}

The second equalities in (\ref{qq2})\ and (\ref{qq3}) follow by a direct
calculation; note that the right hand sides of these expressions are still
covariant, i.e. they look the same in all coordinate systems on $Q$. In
general, the symbol $\widehat{}$ denotes throughout the article the
$\mathbb{R}$-linear operation of minimal quantization. Thus, the minimal
quantization of Hamiltonians (\ref{2.2new}) is given by the following set of
$n$ self-adjoint operators
\begin{equation}
\widehat{h}_{r}=-\frac{1}{2}\hslash^{2}\nabla_{k}A_{r}^{kj}\nabla_{j}-\frac
{1}{2}i\hslash\sum\limits_{\gamma\in I_{\gamma}}d_{\gamma}\left(  \nabla
_{j}(P_{r}^{(\gamma)})^{j}+(P_{r}^{(\gamma)})^{j}\nabla_{j}\right)
+\sum\limits_{\alpha\in I_{\alpha}}c_{\alpha}V_{r}^{(\alpha)}\text{,
}r=1,\ldots n \label{2.3new}%
\end{equation}
acting in the Hilbert space $\mathcal{H}$ $=L^{2}(Q,\left\vert g\right\vert
^{1/2}dq)$. One can show \cite{Ben2002a,Ben2002b,Book} that in the separation
coordinates $\lambda$ the operators (\ref{2.3new}) take the form%
\begin{equation}%
\begin{array}
[c]{ll}%
\widehat{h}_{r}= & -\frac{1}{2}\hslash^{2}\sum_{j=1}^{n}A_{r}^{jj}\left(
\partial_{j}^{2}-\Gamma_{j}\partial_{j}\right)  -i\hslash\sum\limits_{\gamma
\in I_{\gamma}}\sum_{j=1}^{n}d_{\gamma}\left(  P_{r}^{(\gamma)}\right)
^{j}\left(  \partial_{j}-\frac{1}{2}\Gamma_{j}+\frac{1}{2}(\gamma
-m)\lambda_{j}^{-1}\right)  +\\
& +\sum\limits_{\alpha\in I_{\alpha}}c_{\alpha}V_{r}^{(\alpha)}\text{,
\ \ \ \ }r=1,\ldots n\text{, \ \ \ }\partial_{j}=\frac{\partial}%
{\partial\lambda_{j}}%
\end{array}
\label{2.5}%
\end{equation}
where $\Gamma_{j}$ are so called metrically contracted Christoffel symbols of
$g$, expressed in the orthogonal coordinates $\lambda_{i}$ by
\[
\Gamma_{j}=\frac{1}{2}\partial_{j}\ln\frac{\prod_{k\neq j}G^{kk}}{G^{jj}%
}=-\frac{1}{2}m\lambda_{j}^{-1}.
\]
Thus, $\Gamma_{j}$ satisfy the Robertson condition
\begin{equation}
\partial_{k}\Gamma_{j}=0,\ \ \ k\neq j\ \label{R}%
\end{equation}
and so the corresponding eignevalue problems for $\widehat{h}_{r}$
\[
\widehat{h}_{r}\Psi=\varepsilon_{r}\Psi\text{, \ \ \ }r=1,\ldots,n
\]
are multiplicatively separable. It means that they have for each choice of
eigenvalue $\varepsilon_{r}$ of $\widehat{h}_{r}$ the common multiplicatively
separable eigenfunction $\Psi(\lambda_{1},\dotsc,\lambda_{n})=\prod_{i=1}%
^{n}\psi(\lambda_{i})$ with $\psi$ satisfying the following ODE (quantum
separation relation)
\begin{align}
(\varepsilon_{1}\lambda^{n-1}+\varepsilon_{2}\lambda^{n-1}+\dotsb
+\varepsilon_{n})\psi(\lambda)  &  =-\frac{1}{2}\hbar^{2}\lambda^{m}\left(
\frac{d^{2}\psi(\lambda)}{d\lambda^{2}}+\frac{1}{2}m\lambda^{-1}\frac
{d\psi(\lambda)}{d\lambda}\right) \label{sepe}\\
&  +i\hbar\sum\limits_{\gamma\in I_{\gamma}}\lambda_{j}^{\gamma}\left(
\frac{d\psi(\lambda)}{d\lambda}+\frac{1}{4}(\gamma-2m)\lambda^{-1}\psi
(\lambda)\right)  +\sum\limits_{\alpha\in I_{\alpha}}c_{\alpha}\lambda
^{\alpha}\psi(\lambda).\nonumber
\end{align}
The Robertson condition (\ref{R}) implies that $\Gamma_{j}$ satisfy the
pre-Robertson condition \cite{Ben2002a}%
\[
\partial_{k}(\partial_{j}\Gamma_{j}-\frac{1}{2}\Gamma_{j}^{2})=0,\ \ \ k\neq
j
\]
and thus all the Hamiltonian operators (\ref{2.3new}) commute
\begin{equation}
\lbrack\widehat{h}_{r},\widehat{h}_{s}]=0,\ \ \ \ r,s=1,...,n. \label{2.6}%
\end{equation}

Observe that one can eliminate all the linear in $y$ terms in the algebraic
curve (\ref{2.1new}) through the map%

\begin{equation}
x=x^{\prime},\ \ \ \ y=y^{\prime}+\sum_{\gamma\in I_{\gamma}}d_{\gamma
}x^{\gamma-m},\ \ \ \ j=1,...,n \label{zm}%
\end{equation}
which transforms this curve (\ref{2.1new}) to the algebraic curve without
terms linear in $y$ (below we omit the prime at $x$ and $y$)
\[
\sum\limits_{\alpha\in I_{\alpha}}c_{\alpha}x^{\alpha}+\frac{1}{2}\sum
_{\gamma,\gamma^{\prime}\in I_{\gamma}}d_{\gamma}d_{\gamma^{\prime}}%
x^{\gamma+\gamma^{\prime}-m}+\sum_{r=1}^{n}h_{r}x^{n-r}=\frac{1}{2}x^{m}%
y^{2},\ \ \ \ \ \ \ m,\alpha,\gamma\in\mathbb{Z}%
\]
which leads to the Hamiltonians (\ref{2.2new}) in the new form (in the new
variables) without the vector potential terms:
\begin{equation}
\overline{h}_{r}=\frac{1}{2}\mu^{T}A_{r}(\lambda)\mu+\sum\limits_{\alpha\in
I_{\alpha}}c_{\alpha}x^{\alpha}+\frac{1}{2}\sum_{\gamma,\gamma^{\prime}\in
I_{\gamma}}d_{\gamma}d_{\gamma^{\prime}}V_{r}^{(\gamma+\gamma^{\prime}%
-m)},\quad r=1,\dotsc,n. \label{NF}%
\end{equation}
The transformation (\ref{zm}) induces the following canonical transformation
on $\mathcal{M}$:%

\[
\lambda_{j}=\lambda_{j}^{\prime},\ \ \ \ \mu_{j}=\mu_{j}^{\prime}+\sum
_{\gamma\in I_{\gamma}}d_{\gamma}\lambda_{j}^{\gamma-m},\ \ \ \ j=1,...,n
\]
which transforms the Hamiltonians (\ref{2.2new}) to the form (\ref{NF}).

The minimal quantization of (\ref{NF}) yields the operators (given here in
$\lambda$ variables)
\begin{equation}
\widehat{\overline{h}}_{r}=-\frac{1}{2}\hslash^{2}\sum_{j=1}^{n}A_{r}%
^{jj}\left(  \partial_{j}^{2}-\Gamma_{j}\partial_{j}\right)  +\sum
\limits_{\alpha\in I_{\alpha}}c_{\alpha}V_{r}^{(\alpha)}+\frac{1}{2}%
\sum_{\gamma,\gamma^{\prime}\in I_{\gamma}}d_{\gamma}d_{\gamma^{\prime}}%
V_{r}^{(\gamma+\gamma^{\prime}-m)},\quad r=1,\dotsc,n,\text{ \ \ \ }%
\partial_{j}=\frac{\partial}{\partial\lambda_{j}} \label{2.9}%
\end{equation}
acting in the same Hilbert space $\mathcal{H}$ $=L^{2}(Q,\left\vert
g\right\vert ^{1/2}d\lambda)$ as the operators (\ref{2.5}). Both sets of
operators, (\ref{2.5}) and (\ref{2.9}), are related by the quantum canonical
transformation
\begin{equation}
\widehat{\overline{h}}_{r}=U\,\widehat{h}_{r}\,U^{\dagger}%
,\ \ \ \ \ U=U(\lambda)=e^{F(\lambda)}\text{, }F(\lambda)=\left(  -\frac
{i}{\hslash}\sum_{\gamma\in I_{\gamma}}\frac{d_{\gamma}}{\gamma-m+1}%
\sum\limits_{j=1}^{n}\lambda_{j}^{\gamma-m+1}\right)  . \label{QT}%
\end{equation}
Note that the transformation (\ref{QT}) is covariant i.e. is valid in any
coordinate system $q$ on $Q$. Note also that there is no singularity in
(\ref{QT}) as in case $\gamma=m-1\,$\ the corresponding part of the unitary
operation $U\,\widehat{h}_{r}\,U^{\dagger}$ is formally correctly defined.

\section{Classical Painlev\'{e}-type systems\label{sec 2}}

In Part I we constructed Frobenius integrable non-autonomous Hamiltonian
systems with ordinary potentials, generated by the algebraic curve%

\begin{equation}
\sum_{\alpha=-m}^{2n-m+2}c_{\alpha}(t)x^{\alpha}+\sum_{r=1}^{n}h_{r}%
x^{n-r}=\frac{1}{2}x^{m}y^{2},\ \ \ \ m\in\{0,...,n+1\},\quad\label{sA}%
\end{equation}
so that it has the form (\ref{2.1new}) with $I_{\alpha}=\left\{
-m,\ldots,2n-m+2\right\}  $, $I_{\gamma}=\emptyset$ but with $c_{\alpha
}(t)=c_{\alpha}(t_{1},\ldots,t_{n})$ no longer constant but some, for now
undefined, smooth functions of all times $t_{r}$. Solving the corresponding
separation relations yields $n$ Hamiltonians
\begin{equation}
h_{r}=\frac{1}{2}\mu^{T}A_{r}\mu+\sum_{\alpha=-m}^{2n-m+2}c_{\alpha}%
(t)V_{r}^{(\alpha)},\quad r=1,\dotsc,n, \label{1aa}%
\end{equation}
on the phase space $\mathcal{M}$, with the $(2,0)$-tensors $A_{r}$ and the
basic separable potentials given by (\ref{2.4}). Following the method
developed in \cite{arxiv} and in Part I, we now perturb the Hamiltonians
$h_{r}$ to the Hamiltonians
\begin{equation}
h_{1}^{A}=h_{1},\ \ \ \ h_{r}^{A}=h_{r}+W_{r},\ \ \ \ \text{\ }r=2,\ldots,n,
\label{hA}%
\end{equation}
where $W_{r}$ are linear in momenta terms (quasi-St\"{a}ckel terms) given in
$(q,p)$ coordinates by%
\[
W_{r}=\sum\limits_{j=1}^{n}J_{r}^{j}p_{j}%
\]
where $J_{r}$, $r=2,\ldots,n$, are $n-1$ vector fields on $Q$ given by%
\begin{align}
J_{r}  &  =\sum\limits_{k=n+2-m-r}^{n-m}(n+1-m-k)q_{m+r-n-2+k}\partial
_{k}\text{, \ \ \ \ }r\in I_{1}^{m},\label{Zd}\\
J_{r}  &  =-\sum\limits_{k=n+2-m}^{2n+2-m-r}(n+1-m-k)q_{m+r-n-2+k}\partial
_{k}\text{, \ \ \ \ }r\in I_{2}^{m}, \label{Zg}%
\end{align}
(note that they do depend on $m$, as $G$ does) that are Killing vector fields
for the metric $G$. Here and in what follows $\partial_{i}=\partial/\partial
q_{i}$. The index sets $I_{1}^{m}$ and $I_{1}^{m}$ are defined as follows:
\begin{equation}
I_{1}^{m}=\{2,\ldots,n-m+1\},\qquad I_{2}^{m}=\{n-m+2,\ldots,n\},\qquad
m=0,\ldots,n+1. \label{I}%
\end{equation}
with the degenerations $I_{2}^{m}=\emptyset$ for $m=0,1$ while $I_{1}%
^{m}=\emptyset$ for $m=n,n+1$. Here and throughout the paper we use the
notation $q_{0}=1$ and $q_{r}=0$ for $r<0$ and for $r>n\,$. We also set
$W_{1}=0$.

This particular choice of Killing vector fields $J_{r}$ from the whole algebra
of Killing vector fields for $G$ is motivated by the following important
observation. It can be shown \emph{\cite{mb}} that the functions
$\mathcal{E}_{r}=E_{r}+W_{r}$ (called the \emph{geodesic quasi-St\"{a}ckel
Hamiltonians}) span the Lie algebra $\mathfrak{g}=\mathrm{span}\{\mathcal{E}%
_{r},\ $ $r=1,\ldots,n\}$ with the commutation relations
\[
\{\mathcal{E}_{1},\mathcal{E}_{r}\}=0,\quad r=2,\dots,n,
\]
and
\begin{equation}
\{\mathcal{E}_{r},\mathcal{E}_{s}\}=%
\begin{cases}
0 & \text{for $r\in I_{1}^{m}$ and $s\in I_{2}^{m}$},\\
(s-r)\mathcal{E}_{r+s-(n-m+2)} & \text{for $r,s\in I_{1}^{m}$},\\
-(s-r)\mathcal{E}_{r+s-(n-m+2)} & \text{for $r,s\in I_{2}^{m}$},
\end{cases}
\label{str}%
\end{equation}
where $\mathcal{E}_{i}=0$ as soon as $i\leq0$ or $i>n$. The algebra
$\mathfrak{g}$ has an Abelian subalgebra
\begin{equation}
\mathfrak{a}=\mathrm{span}\left\{  \mathcal{E}_{1},\dotsc,\mathcal{E}%
_{\kappa_{1}},\mathcal{E}_{n-\kappa_{2}+1},\dotsc,\mathcal{E}_{n}\right\}
\label{gma}%
\end{equation}
where%
\begin{equation}
\kappa_{1}=\left[  \frac{n+3-m}{2}\right]  ,\qquad\kappa_{2}=\left[  \frac
{m}{2}\right]  . \label{kappy}%
\end{equation}
Finally, let us construct new Hamiltonians $H_{r}^{A}$ such that for
$r\in\{1\}\cup I_{1}^{m}$
\begin{align}
H_{r}^{A}  &  =h_{r}^{A},\quad\text{for $r=1,\dotsc,\kappa_{1}$},\nonumber\\
H_{r}^{A}  &  =\sum_{j=1}^{r}\zeta_{r,j}(t_{1},\dotsc,t_{r-1})h_{j}^{A}%
,\quad\zeta_{r,r}=1,\quad\text{for $r=\kappa_{1}+1,\dotsc,n-m+1$} \label{7b}%
\end{align}
and for $r\in I_{2}^{m}$
\begin{align}
H_{r}^{A}  &  =\sum_{j=0}^{n-r}\zeta_{r,r+j}(t_{r+1},\dotsc,t_{n})h_{r+j}%
^{A},\quad\zeta_{r,r}=1,\quad\text{for $r=n-m+2,\dotsc,n-\kappa_{2}$%
},\nonumber\\
H_{r}^{A}  &  =h_{r}^{A},\quad\text{for $r=n-\kappa_{2}+1,\dotsc,n$},
\label{7c}%
\end{align}
where $\zeta_{r,j}(t)$ are some functions of appropriate evolution parameters.
The Hamiltonians $H_{r}^{A}$ define $n$ non-autonomous Hamiltonian systems on
$\mathcal{M}$
\begin{equation}
\xi_{t_{r}}=Y_{r}(\xi,t)=\pi dH_{r}^{A}(\xi,t),\qquad r=1,\dots,n \label{nhs}%
\end{equation}
and according to results obtained in \cite{arxiv} and in Part I one can always
obtain (by solving an appropriate compatible and overdetermined system of
PDE's) the explicit form of functions $c_{\alpha}(t_{1},\ldots,t_{n})$,
$\zeta_{r,j}(t_{1},\dotsc,t_{r-1})$ and $\zeta_{r,r+j}(t_{r+1},\dotsc,t_{n})$
such that the Hamiltonians $H_{r}^{A}$ in (\ref{7b}) and (\ref{7c}) satisfy
the (classical) Frobenius integrability condition%
\begin{equation}
\frac{\partial H_{r}}{\partial t_{s}}-\frac{\partial H_{s}}{\partial t_{r}%
}+\{H_{r},H_{s}\}=f_{rs}(t_{1},\ldots,t_{n}),\quad r,s=1,\dots,n \label{fcg}%
\end{equation}
with $H_{r}=H_{r}^{A}$. It means that $n$ Hamiltonian systems (\ref{nhs}) have
(at least for small intervals of times $t_{i}$) a unique common multi-time
solution $\xi=\xi(t_{1},\dots,t_{n},\xi_{0})$ for any initial condition
$\xi_{0}$ \cite{Fecko,Lundell}. In consequence, the non-autonomous Hamiltonian
vector fields $Y_{r}\ $on $\mathcal{M}$ satisfy the Frobenius condition%
\begin{equation}
\frac{\partial Y_{r}}{\partial t_{s}}-\frac{\partial Y_{s}}{\partial t_{r}%
}-\left[  Y_{r},Y_{s}\right]  =0\text{ for all}\ r,s=1,\ldots,n \label{genFr}%
\end{equation}
(see Part I for the details of this construction; the second minus sign in
(\ref{genFr}) is due to the convention used in (\ref{konwencja})).

In Part I we also constructed Frobenius integrable non-autonomous Hamiltonian
systems with vector potentials, generated by the following algebraic curve%

\begin{equation}
\sum_{\gamma=0}^{n+1}d_{\gamma}(t)x^{\gamma}y+\sum_{r=1}^{n}x^{n-r}h_{r}%
=\frac{1}{2}x^{m}y^{2}, \label{mag}%
\end{equation}
i.e. the curve (\ref{2.1new}) with $I_{\alpha}=\emptyset$, $I_{\gamma
}=\left\{  0,\ldots,n+1\right\}  $ and with $d_{\gamma}(t)=d_{\gamma}%
(t_{1},\dotsc,t_{n})$ no longer constant but, for the moment arbitrary,
functions of all times $t_{i}$. Solving the corresponding separation relations
with respect to $h_{r}$ yields $n$ Hamiltonians
\begin{equation}
h_{r}=\frac{1}{2}\mu^{T}A_{r}\mu+\sum_{\gamma=0}^{n+1}d_{\gamma}(t)\mu
^{T}P_{r}^{(\gamma)},\quad r=1,\dotsc,n \label{mag1}%
\end{equation}
on $\mathcal{M}=T^{\ast}Q$, with the contravariant tensors $A_{r}$ and with
the vector potentials $P_{r}^{(\gamma)}$ given by (\ref{2.4}) that enter the
Hamiltonians (\ref{mag1}) through the (linear in momenta)\ magnetic terms%
\begin{equation}
M_{r}^{(\gamma)}=\mu^{T}P_{r}^{(\gamma)}. \label{MT}%
\end{equation}
Deforming the Hamiltonians (\ref{mag1}) to the Hamiltonians $h_{r}^{B}%
=h_{r}+W_{r}$ (cf. (\ref{hA})) using the same quasi-St\"{a}ckel terms
$W_{r}=J_{r}^{j}p_{j}$ with the same vector fields $J_{r}$ given by
(\ref{Zd})\ and (\ref{Zg}) we were able to construct, through an appropriate
choice of the functions $\zeta_{r}$ in (\ref{7b}), (\ref{7c}), as well as an
appropriate choice of $d_{\gamma}(t)$ in (\ref{mag1}), the set of $n$
non-autonomous Hamiltonians $H_{r}^{B}(t)$, $r=1,...,n,$ satisfying the
Frobenius integrability condition (\ref{fcg}).

In Part II\ we have constructed the isomonodromic Lax representation for the
non-homogeneous systems $\xi_{t_{r}}=\pi dH_{r}^{A}(\xi,t)$ and $\xi_{t_{r}%
}=\pi dH_{r}^{B}(\xi,t)$ thus proving that the are indeed of Painlev\'{e}-type.

\section{Quantization of Painlev\'{e}-type systems\label{s4}}

As we explained in Introduction, the aim of this paper is to show that the
deformation procedures, developed in \cite{arxiv} and in Part I, and shortly
presented in the previous section, have its quantum counterpart. In this
section we prove that the minimally quantized Painlev\'{e}-type systems, both
with ordinary and with magnetic potentials, satisfy the quantum Frobenius
condition (\ref{fcq}).

\subsection{Quantum Frobenius condition}

Consider the set of Schr\"{o}dinger equations%
\begin{equation}
i\hslash\frac{\partial\Psi}{\partial t_{r}}=\widehat{H}_{r}\Psi\text{,
\ }r=1,\ldots,n \label{Sch}%
\end{equation}
where $\widehat{H}_{r}$ is a set of $n$ linear operators acting in a Hilbert
space $\mathcal{H}$. A necessary condition for the existence of a common
multi-time solution $\Psi(t_{1},\ldots,t_{n})$ of the system (\ref{Sch}) has
the form
\[
\frac{\partial^{2}\Psi}{\partial t_{r}\partial t_{s}}=\frac{\partial^{2}\Psi
}{\partial t_{s}\partial t_{r}}\text{ for any }r,s=1,\ldots,n.
\]
Inserting it into (\ref{Sch}) leads to the following necessary condition for
the existence of common solutions of (\ref{Sch})
\begin{equation}
i\hslash\frac{\partial\widehat{H}_{r}}{\partial t_{s}}-i\hslash\frac
{\partial\widehat{H}_{s}}{\partial t_{r}}+\left[  \widehat{H}_{r}%
,\widehat{H}_{s}\right]  =0\text{, \ }r,s=1,\ldots,n. \label{fcq}%
\end{equation}
We will refer to the condition (\ref{fcq}) as the quantum Frobenius condition.
We stress that while the right hand sides $f_{rs}$ of the classical Frobenius
condition (\ref{fcg}) may in general depend on all the times $t_{i}$, the
right hand side of (\ref{fcq}) must be zero. This however does not restrict
our method, described below, as it is always possible in the classical regime
to choose the functions $c_{\alpha}$\ so that $f_{rs}(t)=0$ (while in the
magnetic case $f_{rs}$ are always zero).

\subsection{Quantization of systems with ordinary potentials}

We start with the case of ordinary potentials. According to Definition
\ref{minq}, the minimal quantization $\widehat{H}_{r}^{A}$\ of Hamiltonians
$H_{r}^{A}$ is obtained by replacing the Hamiltonians $h_{r}^{A}$ in
(\ref{7b}) and (\ref{7c}) by the self-adjoint operators
\begin{equation}
\widehat{h}_{r}^{A}=\widehat{E}_{r}+\widehat{W}_{r}+\sum\limits_{\alpha
=-m}^{2n-m+2}c_{\alpha}(t_{1},\ldots,t_{n})V_{r}^{(\alpha)}%
,\ \ \ \ \ r=1,...,n, \label{q1}%
\end{equation}
where, due to (\ref{qq2}) and (\ref{qq3})
\begin{align}
\widehat{E}_{r}  &  =-\frac{1}{2}\hslash^{2}\nabla_{k}A_{r}^{kj}\nabla
_{j}=-\frac{1}{2}\hslash^{2}\left\vert g\right\vert ^{-\frac{1}{2}}%
\partial_{k}\left\vert g\right\vert ^{\frac{1}{2}}A_{r}^{kj}\partial
_{j}\label{q2}\\
\widehat{W}_{r}  &  =-\frac{1}{2}i\hslash\left(  \nabla_{j}J_{r}^{j}+J_{r}%
^{j}\nabla_{j}\right)  =-\frac{1}{2}i\hslash\left(  \left\vert g\right\vert
^{-\frac{1}{2}}\partial_{j}\left\vert g\right\vert ^{\frac{1}{2}}J_{r}%
^{j}+J_{r}^{j}\partial_{j}\right)  \label{q3}%
\end{align}
are operators acting in the Hilbert space $\mathcal{H}$ $=L^{2}(Q,\left\vert
g\right\vert ^{1/2}dq)$.

\begin{theorem}
\label{MAIN}For any fixed $m\in\left\{  0,\ldots,n+1\right\}  $, the set of
$n$ operators $\widehat{H}_{r}^{A}$, $r=1,\ldots n$, given by (\ref{7b}),
(\ref{7c}) with $h_{r}^{A}$ replaced by \ $\widehat{h}_{r}^{A}$ given by
(\ref{q1}), satisfies the quantum Frobenius condition (\ref{fcq}).
\end{theorem}

We will prove this theorem by showing that the proof of the classical version
of this theorem, as it is presented in Part I, survives in the quantum regime.
By results in \cite{arxiv} and in Part I, also mentioned in Section
\ref{sec 2}, the Hamiltonians $\widehat{H}_{r}^{A}$ will satisfy the quantum
Frobenius condition (\ref{fcq}) as soon as any pair of operators
$\widehat{\mathcal{E}}_{r}+V_{r}^{(\alpha)}\equiv\widehat{E}_{r}%
+\widehat{W}_{r}+V_{r}^{(\alpha)}$ and $\widehat{\mathcal{E}}_{s}%
+V_{s}^{(\alpha)}\equiv\widehat{E}_{s}+\widehat{W}_{s}+V_{s}^{(\alpha)}$ will,
for any fixed $m\in\left\{  0,\ldots,n+1\right\}  $, satisfy the same, up to
the factor $i\hslash$, commutation relations as their classical counterparts
$\mathcal{E}_{r}+V_{r}^{(\alpha)}\equiv E_{r}+W_{r}+V_{r}^{(\alpha)}$ and
$\mathcal{E}_{s}+V_{s}^{(\alpha)}\equiv E_{s}+W_{s}+V_{s}^{(\alpha)}$. That
is, we have to prove that
\begin{equation}
i\hslash\widehat{\left\{  E_{r}+W_{r}+V_{r}^{(\alpha)},E_{s}+W_{s}%
+V_{s}^{(\alpha)}\right\}  }=\left[  \widehat{E}_{r}+\widehat{W}_{r}%
+V_{r}^{(\alpha)},\widehat{E}_{s}+\widehat{W}_{s}+V_{s}^{(\alpha)}\right]
\label{wk1}%
\end{equation}
for all $r,s=1,\ldots,n$, for all $m\in\left\{  0,\ldots,n+1\right\}  $ and
for all $\alpha\in\left\{  -m,\ldots,2n-m+2\right\}  $, as the relation
(\ref{wk1}) means that we can follow the procedure from Section \ref{sec 2}
also in the quantum case and this procedure will yield the same PDE's for
functions $c_{\alpha}(t_{1},\ldots,t_{n})$, $\zeta_{r,j}(t_{1},\dotsc
,t_{r-1})$ and $\zeta_{r,r+j}(t_{r+1},\dotsc,t_{n})$ as in the corresponding
classical procedure. We will perform all calculations in Vi\`{e}te coordinates
$(q,p)$.

Before we proceed, let us rewrite (\ref{q2}) and (\ref{q3}) with the operators
$\partial_{j}$ standing maximally to the right. Note that due to (\ref{n2}) we
have $\det L=(-1)^{n}q_{n}$ and $\det G_{0}=(-1)^{[n/2]}$ (where $[\cdot]$
denotes the integer part) so that, due to (\ref{n1})%

\begin{equation}
\left\vert g\right\vert =\varepsilon q_{n}^{-m}\text{ with }\varepsilon
=(-1)^{nm+[n/2]}\text{.} \label{q4}%
\end{equation}
Thus, $\partial_{j}\left\vert g\right\vert =-\varepsilon mq_{n}^{-m-1}%
\delta_{j,n}$ and a direct calculation yields that
\[
\widehat{E}_{r}=-\frac{1}{2}\hslash^{2}\left(  A_{r}^{kj}\partial_{k}%
\partial_{j}+(A_{r}^{kj})_{q_{k}}\partial_{j}-\frac{1}{2}\varepsilon^{2}%
mq_{n}^{-1}A_{r}^{nj}\partial_{j}\right)
\]
while%
\[
\widehat{W}_{r}=-i\hslash\left(  J_{r}^{j}\partial_{j}+\frac{1}{2}\left\vert
g\right\vert ^{-1/2}\left(  \partial_{j}\left\vert g\right\vert ^{1/2}%
J_{r}^{j}\right)  \right)  .
\]
A straightforward calculation shows that the zero-order term in $\widehat{W}%
_{r}$ is equal to $0$ for all $r$. Indeed, due to (\ref{q4}) we have%
\[
\left\vert g\right\vert ^{-1/2}\left(  \partial_{j}\left\vert g\right\vert
^{1/2}J_{r}^{j}\right)  =-\frac{1}{2}mq_{n}^{-1}J_{r}^{n}+\left(  \partial
_{j}J_{r}^{j}\right)  .
\]
Assume first that $r\in I_{1}^{m}$. Then, due to (\ref{Zd}), $J_{r}^{n}=0$ as
soon as $m>0$ and thus $mq_{n}^{-1}J_{r}^{n}$ is always (for all $m$) zero.
Further, again due to (\ref{Zd})%

\[
\left(  \partial_{j}J_{r}^{j}\right)  =(n-m-j+1)\delta_{r,n-m+2}=0\text{ (no
summation)}%
\]
due to definition of $I_{1}^{m}$. A similar, albeit a little bit more tedious,
calculation shows that the same result holds for any $r\in I_{2}^{m} $. In
consequence the Hamiltonian operators (\ref{q2}) and (\ref{q3}) can be written
as
\begin{align}
\widehat{E}_{r}  &  =-\frac{1}{2}\hslash^{2}\left(  A_{r}^{kj}\partial
_{k}\partial_{j}+(A_{r}^{kj})_{q_{k}}\partial_{j}-\frac{1}{2}mq_{n}^{-1}%
A_{r}^{nj}\partial_{j}\right) \nonumber\\
& \label{q8}\\
\widehat{W}_{r}  &  =-i\hslash J_{r}^{j}\partial_{j}=-i\hslash J_{r}.\nonumber
\end{align}
We are now in position to prove the condition (\ref{wk1}). Obviously, for any
$r,s\in\left\{  1,\ldots,n\right\}  $%
\begin{align*}
\left\{  E_{r}+W_{r}+V_{r}^{(\alpha)},E_{s}+W_{s}+V_{s}^{(\alpha)}\right\}
&  =\left\{  \mathcal{E}_{r},\mathcal{E}_{s}\right\}  +\left\{  E_{r}%
,V_{s}^{(\alpha)}\right\}  +\left\{  V_{r}^{(\alpha)},E_{s}\right\} \\
&  +\left\{  W_{r},V_{s}^{(\alpha)}\right\}  +\left\{  V_{r}^{(\alpha)}%
,W_{s}\right\}  +\left\{  V_{r}^{(\alpha)},V_{s}^{(\alpha)}\right\}  .
\end{align*}
The first term $\left\{  \mathcal{E}_{r},\mathcal{E}_{s}\right\}  $ is given
by (\ref{str}). The last term $\left\{  V_{r}^{(\alpha)},V_{s}^{(\alpha
)}\right\}  $ is obviously $0$. Moreover, St\"{a}ckel Hamiltonians themselves
(i.e. without the quasi-St\"{a}ckel term $W_{r}$), both geodesic and with
potentials, commute with each other, so that $\left\{  E_{r},E_{s}\right\}
=0$ and $\left\{  E_{r}+V_{r}^{(\alpha)},E_{s}+V_{s}^{(\alpha)}\right\}  =0$,
which yields%
\[
\left\{  E_{r},V_{s}^{(\alpha)}\right\}  +\left\{  V_{r}^{(\alpha)}%
,E_{s}\right\}  =0.
\]
Thus, the left hand side of (\ref{wk1}) becomes
\begin{equation}
i\hslash\widehat{\left\{  \mathcal{E}_{r},\mathcal{E}_{s}\right\}  }%
+i\hslash\widehat{\left(  \left\{  W_{r},V_{s}^{(\alpha)}\right\}  +\left\{
V_{r}^{(\alpha)},W_{s}\right\}  \right)  } \label{lewa}%
\end{equation}
On the other hand, the right hand side of (\ref{wk1}) is%

\[
\lbrack\widehat{\mathcal{E}}_{r},\widehat{\mathcal{E}}_{s}]+\left[
\widehat{E}_{r},V_{s}^{(\alpha)}\right]  +\left[  V_{r}^{(\alpha)}%
,\widehat{E}_{s}\right]  +\left[  \widehat{W}_{r},V_{s}^{(\alpha)}\right]
+\left[  V_{r}^{(\alpha)},\widehat{W}_{s}\right]  +\left[  V_{r}^{(\alpha
)},V_{s}^{(\alpha)}\right]
\]
(also here the last term is $0$) where the operators $\widehat{\mathcal{E}%
}_{r}$ denote minimal quantization of respective geodesic quasi-St\"{a}ckel
Hamiltonians $\mathcal{E}_{r}$%
\begin{equation}
\widehat{\mathcal{E}}_{r}=\widehat{E}_{r}+\widehat{W}_{r},\ r=1,\ldots,n.
\label{hq}%
\end{equation}
We moreover know that the quantum St\"{a}ckel Hamiltonians themselves (i.e.
without the quasi-St\"{a}ckel term $\widehat{W}_{r}$), both geodesic and with
potentials, commute with each other, due to (\ref{2.6}), i.e. $\left[
\widehat{E}_{r},\widehat{E}_{s}\right]  =0$ and $\left[  \widehat{E}_{r}%
+V_{r}^{(\alpha)},\widehat{E}_{s}+V_{s}^{(\alpha)}\right]  =0$. This also
yields%
\[
\left[  \widehat{E}_{r},V_{s}^{(\alpha)}\right]  +\left[  V_{r}^{(\alpha
)},\widehat{E}_{s}\right]  =0
\]
and so the right hand side of (\ref{wk1}) becomes%
\begin{equation}
\lbrack\widehat{\mathcal{E}}_{r},\widehat{\mathcal{E}}_{s}]+\left[
\widehat{W}_{r},V_{s}^{(\alpha)}\right]  +\left[  V_{r}^{(\alpha)}%
,\widehat{W}_{s}\right]  . \label{prawa}%
\end{equation}
Further
\[
\{W_{r},V_{s}^{(\alpha)}\}+\{V_{r}^{(\alpha)},W_{s}\}=-J_{r}^{j}%
(V_{s}^{(\alpha)})_{q_{j}}+J_{s}^{j}(V_{r}^{(\alpha)})_{q_{j}},
\]
is thus an expression on $Q$ not on $\mathcal{M}$ i.e. a function on $Q$,
while on the quantum level
\begin{align*}
\lbrack\widehat{W}_{r},V_{s}^{(\alpha)}]+[V_{r}^{(\alpha)},\widehat{W}_{s}]
&  =-i\hslash\left(  \lbrack J_{r}^{j}\partial_{j},V_{s}^{(\alpha)}%
]+[V_{r}^{(\alpha)},J_{s}^{j}\partial_{j}]\right) \\
&  =-i\hslash\left(  J_{r}^{j}(V_{s}^{(\alpha)})_{q_{j}}-J_{s}^{j}%
(V_{r}^{(\alpha)})_{q_{j}}\right) \\
&  =i\hslash\left(  \{W_{r},V_{s}^{(\alpha)}\}+\{V_{r}^{(\alpha)}%
,W_{s}\}\right)
\end{align*}
and thus (see Definition \ref{minq})%

\[
i\hslash\widehat{\left(  \{W_{r},V_{s}^{(\alpha)}\}+\{V_{r}^{(\alpha)}%
,W_{s}\}\right)  }=[\widehat{W}_{r},V_{s}^{(\alpha)}]+[V_{r}^{(\alpha
)},\widehat{W}_{s}].
\]
It remains to show that $i\hslash\widehat{\left\{  \mathcal{E}_{r}%
,\mathcal{E}_{s}\right\}  }=[\widehat{\mathcal{E}}_{r},\widehat{\mathcal{E}%
}_{s}]$ (that actually shows that the operators $\widehat{\mathcal{E}}_{r}$
constitute a Lie algebra with the same, up to the factor $i\hslash$, structure
constants as the algebra (\ref{str}) generated by their classical counterparts
$\mathcal{E}_{r}=E_{r}+W_{r}$).

\begin{theorem}
\label{algebra}The operators $\widehat{\mathcal{E}}_{r}$ in (\ref{hq}) satisfy
the commutation relations%
\[
\lbrack\widehat{\mathcal{E}}_{1},\widehat{\mathcal{E}}_{r}]=0\text{,
\ \ \ }r=2,\ldots,n,
\]%
\[
\lbrack\widehat{\mathcal{E}}_{r},\widehat{\mathcal{E}}_{s}]=%
\begin{cases}
0, & \text{for }r\in I_{1}^{m}\text{ and }s\in I_{2}^{m},\\
i\hslash(s-r)\widehat{\mathcal{E}}_{r+s-(n-m+2)}, & \text{for }r,s\in
I_{1}^{m},\\
-i\hslash(s-r)\widehat{\mathcal{E}}_{r+s-(n-m+2)}, & \text{for }r,s\in
I_{2}^{m}.
\end{cases}
\]
(so that $[\widehat{\mathcal{E}}_{r},\widehat{\mathcal{E}}_{s}]=i\hslash
\widehat{\left\{  \mathcal{E}_{r},\mathcal{E}_{s}\right\}  }$).
\end{theorem}

Let us see why this theorem is true. Naturally%
\begin{equation}
\left\{  \mathcal{E}_{r},\mathcal{E}_{s}\right\}  =\left\{  E_{r}%
,E_{s}\right\}  +\left\{  E_{r},W_{s}\right\}  +\left\{  W_{r},E_{s}\right\}
+\left\{  W_{r},W_{s}\right\}  . \label{komklas}%
\end{equation}
(with $\left\{  E_{r},E_{s}\right\}  =0$) while%
\begin{equation}
\lbrack\widehat{\mathcal{E}}_{r},\widehat{\mathcal{E}}_{s}]=\left[
\widehat{E}_{r}+\widehat{W}_{r},\widehat{E}_{s}+\widehat{W}_{s}\right]
=\left[  \widehat{E}_{r},\widehat{E}_{s}\right]  +\left[  \widehat{E}%
_{r},\widehat{W}_{s}\right]  +\left[  \widehat{W}_{r},\widehat{E}_{s}\right]
+\left[  \widehat{W}_{r},\widehat{W}_{s}\right]  \label{kom}%
\end{equation}
and $\left[  \widehat{E}_{r},\widehat{E}_{s}\right]  =0$. In order to
calculate the remaining terms in (\ref{kom}) we we will use the link between
the canonical Poisson bracket on $\mathcal{M}=T^{\ast}Q$ and the Schouten
bracket between symmetric contravariant tensors on $Q$. For a pair of
functions on $\mathcal{M}$
\[
F_{K}=\frac{1}{k!}K^{i_{1}...i_{k}}(q)p_{i_{1}}...p_{i_{k}},\ \ \ \ F_{R}%
=\frac{1}{r!}R^{i_{1}...i_{r}}(q)p_{i_{1}}...p_{i_{r}}%
\]
where $K$ and $R$ are two symmetric tensors on $Q$, of type $(k,0)$ and
$(r,0)$ respectively, the following relation holds \cite{Dolan}
\begin{equation}
\{F_{K},F_{R}\}=-[K,R]_{S}{}^{i_{1}...i_{k+r-1}}p_{i_{1}}...p_{i_{k+r-1}}
\label{relacja}%
\end{equation}
(the minus sign in (\ref{relacja}) is due to the convention used in
(\ref{konwencja})) where
\begin{equation}
\lbrack K,R]_{S}{}^{l_{1}...l_{k+r-1}}=\frac{1}{k!r!}\left[  kK^{i(l_{1}%
...}\partial_{i}R^{...l_{k+r-1})}-rR^{i(l_{1}...}\partial_{i}K^{...l_{k+r-1}%
)}\right]  \label{S}%
\end{equation}
is a $(k+r-1,0)$-type symmetric tensor on $Q$, called the Schouten bracket of
$K$ and $R$; the symbol $(...)$ on the right hand side of (\ref{S}) denotes
symmetrization over indices. Thus
\[
\{W_{r},W_{s}\}=\left[  J_{s},J_{r}\right]  _{S}^{j}\,p_{j}=\left(  J_{s}%
^{k}(J_{r}^{j})_{q_{k}}-J_{r}^{k}(J_{s}^{j})_{q_{k}}\right)  p_{j}%
\]
(where $(J_{r}^{j})_{q_{k}}$ denotes $\frac{\partial}{\partial_{q_{k}}}%
J_{r}^{j}$) and by comparing (\ref{komklas}) and (\ref{str}) we obtain
\begin{equation}
J_{s}^{k}(J_{r}^{j})_{q_{k}}-J_{r}^{k}(J_{s}^{j})_{q_{k}}=%
\begin{cases}
0, & \text{for }r\in I_{1}^{m}\text{ and }s\in I_{2}^{m},\\
(s-r)J_{r+s-(n-m+2)}^{j}, & \text{for }r,s\in I_{1}^{m},\\
-(s-r)J_{r+s-(n-m+2)}^{j}, & \text{for }r,s\in I_{2}^{m},
\end{cases}
\label{q9}%
\end{equation}
Moreover, again due to (\ref{relacja}),
\[
\{W_{r},E_{s}\}+\{E_{r},W_{s}\}=\left(  \left[  A_{s},J_{r}\right]  _{S}%
^{kj}+\left[  J_{s},A_{r}\right]  _{S}^{kj}\right)  p_{k}p_{j},
\]
so, by comparing the appropriate terms in (\ref{komklas}) and in (\ref{str})
and due to (\ref{S}), the above formula reads
\begin{align}
&  A_{s}^{lk}(J_{r}^{j})_{q_{l}}+A_{s}^{lj}(J_{r}^{k})_{q_{l}}-J_{r}^{l}%
(A_{s}^{kj})_{q_{l}}+J_{s}^{l}(A_{r}^{kj})_{q_{l}}-A_{r}^{lk}(J_{s}%
^{j})_{q_{l}}-A_{r}^{lj}(J_{s}^{k})_{q_{l}}\nonumber\\
& \label{q10}\\
&  =%
\begin{cases}
0, & \text{for }r\in I_{1}^{m}\text{ and }s\in I_{2}^{m},\\
(s-r)A_{r+s-(n-m+2)}^{kj}, & \text{for }r,s\in I_{1}^{m},\\
-(s-r)A_{r+s-(n-m+2)}^{kj}, & \text{for }r,s\in I_{2}^{m}.
\end{cases}
\nonumber
\end{align}
Further, differentiating (\ref{q10}) with respect to $q_{k}$ and summation
over $k$ yields
\begin{align}
&  (A_{s}^{lk})_{q_{k}}(J_{r}^{j})_{q_{l}}-(A_{r}^{lk})_{q_{k}}(J_{s}%
^{j})_{q_{l}}-J_{r}^{l}(A_{s}^{kj})_{q_{l}q_{k}}+J_{s}^{l}(A_{r}^{kj}%
)_{q_{l}q_{k}}\nonumber\\
& \label{q11}\\
&  =%
\begin{cases}
0, & \text{for }r\in I_{1}^{m}\text{ and }s\in I_{2}^{m},\\
(s-r)(A_{r+s-(n-m+2)}^{kj})_{q_{k}}, & \text{for }r,s\in I_{1}^{m},\\
-(s-r)(A_{r+s-(n-m+2)}^{kj})_{q_{k}}, & \text{for }r,s\in I_{2}^{m},
\end{cases}
\nonumber
\end{align}
due to fact that all components of all $J_{r}$ are linear in $q$. Using the
above classical relations we can now calculate the remaining terms in
(\ref{kom}).

\begin{lemma}
\label{l2}The commutator of differential operators $\widehat{W}_{r}$ and
$\widehat{W}_{s}$ is
\begin{equation}
\lbrack\widehat{W}_{r},\widehat{W}_{s}]=%
\begin{cases}
0, & \text{for }r\in I_{1}^{m}\text{ and }s\in I_{2}^{m},\\
i\hslash(s-r)\widehat{W}_{r+s-(n-m+2)}, & \text{for }r,s\in I_{1}^{m},\\
-i\hslash(s-r)\widehat{W}_{r+s-(n-m+2)}, & \text{for }r,s\in I_{2}^{m}.
\end{cases}
\label{q12}%
\end{equation}
(so that $[\widehat{W}_{r},\widehat{W}_{s}]=i\hslash\widehat{\{W_{r},W_{s}\}}$).
\end{lemma}

The proof is by direct computation:
\begin{align*}
\lbrack\widehat{W}_{r},\widehat{W}_{s}]  &  =-\hslash^{2}[J_{r},J_{s}%
]=-\hslash^{2}\left(  J_{r}^{k}(J_{s}^{j})_{q_{k}}-J_{s}^{k}(J_{r}^{j}%
)_{q_{k}}\right)  \partial_{j}\\
&  \overset{(\ref{q9})}{=}%
\begin{cases}
0, & \text{for }r\in I_{1}^{m}\text{ and }s\in I_{2}^{m},\\
\hslash^{2}(s-r)J_{r+s-(n-m+2)}^{j}, & \text{for }r,s\in I_{1}^{m},\\
-\hslash^{2}(s-r)J_{r+s-(n-m+2)}^{j}, & \text{for }r,s\in I_{2}^{m},
\end{cases}
\end{align*}
which yields immediately (\ref{q12}) as $\widehat{W}_{r}=-i\hslash J_{r}$.

\begin{lemma}
\label{l3}The mixed term in (\ref{kom}) is given by%
\begin{equation}
\lbrack\widehat{W}_{r},\widehat{E}_{s}]+[\widehat{E}_{r},\widehat{W}_{s}]=%
\begin{cases}
0, & \text{for }r\in I_{1}^{m}\text{ and }s\in I_{2}^{m},\\
i\hslash(s-r)\widehat{E}_{r+s-(n-m+2)}, & \text{for }r,s\in I_{1}^{m},\\
-i\hslash(s-r)\widehat{E}_{r+s-(n-m+2)}, & \text{for }r,s\in I_{2}^{m}.
\end{cases}
\label{q13}%
\end{equation}

\end{lemma}

\begin{proof}
Assume $r,s\in I_{1}^{m}$ (for $r,s\in I_{2}^{m}$ the proof is analogous).
Using (\ref{q8}) we obtain
\begin{align}
\lbrack\widehat{W}_{r},\widehat{E}_{s}]+[\widehat{E}_{r},\widehat{W}_{s}]  &
=\frac{1}{2}i\hslash^{3}\{[J_{r}^{l}\partial_{l},A_{s}^{kj}\partial
_{k}\partial_{j}]+[A_{r}^{kj}\partial_{k}\partial_{j},J_{s}^{l}\partial
_{l}]\nonumber\\
&  +[J_{r}^{l}\partial_{l},(A_{s}^{kj})_{q_{k}}\partial_{j}]+[(A_{r}%
^{kj})_{q_{k}}\partial_{j},J_{s}^{l}\partial_{l}]\label{zbieraj}\\
&  -\frac{1}{2}m([J_{r}^{l}\partial_{l},q_{n}^{-1}A_{s}^{nj}\partial
_{j}]+[q_{n}^{-1}A_{r}^{nj}\partial_{j},J_{s}^{l}\partial_{l}])\}\nonumber
\end{align}
Since $(J_{r}^{l})_{q_{j}q_{k}}=0$, the first two terms above read
\begin{align*}
&  [J_{r}^{l}\partial_{l},A_{s}^{kj}\partial_{k}\partial_{j}]+[A_{r}%
^{kj}\partial_{k}\partial_{j},J_{s}^{l}\partial_{l}]\overset{(\ref{S}%
)}{=}([J_{r},A_{s}]_{S}^{kj}+[A_{r},J_{s}]_{S}^{kj})\partial_{k}\partial_{j}\\
&  \overset{(\ref{q10})}{=}-(s-r)A_{r+s-(n-m+2)}^{kj}\partial_{k}\partial_{j}.
\end{align*}
while the next two terms become%
\begin{align*}
&  [J_{r}^{l}\partial_{l},(A_{s}^{kj})_{q_{k}}\partial_{j}]+[(A_{r}%
^{kj})_{q_{k}}\partial_{j},J_{s}^{l}\partial_{l}]\\
&  =\left(  J_{r}^{l}(A_{s}^{kj})_{q_{k}q_{l}}-(A_{s}^{kl})_{q_{k}}(J_{r}%
^{j})_{q_{l}}+(A_{r}^{kl})_{q_{k}}(J_{s}^{j})_{q_{l}}-J_{s}^{l}(A_{r}%
^{kj})_{q_{k}q_{l}}\right)  \partial_{j}\\
&  \overset{(\ref{q11})}{=}-(s-r)(A_{r+s-(n-m+2)}^{kj})_{q_{k}}\partial_{j}.
\end{align*}
The last two terms in (\ref{zbieraj}) are%
\begin{align*}
-\frac{1}{2}m\left(  [J_{r}^{l}\partial_{l},q_{n}^{-1}A_{s}^{nj}\partial
_{j}]+[q_{n}^{-1}A_{r}^{nj}\partial_{j},J_{s}^{l}\partial_{l}]\right)   &
=-\frac{1}{2}mq_{n}^{-1}([J_{r}^{l}\partial_{l},A_{s}^{nj}\partial_{j}%
]+[A_{r}^{nj}\partial_{j},J_{s}^{l}\partial_{l}])\\
&  +\frac{1}{2}mq_{n}^{-2}(J_{r}^{n}A_{s}^{nj}-J_{s}^{n}A_{r}^{nj}%
)\partial_{j}%
\end{align*}
and as for $m>0$%
\[
J_{r}^{n}=0,\ \ \ \ r\neq n-m+2\text{ \ and \ }J_{n-m+2}^{n}=(m-1)q_{n}%
\]
we can formally write
\[
\frac{1}{q_{n}}J_{r}^{n}A_{s}^{nj}=(J_{r}^{n})_{q_{n}}A_{s}^{nj}=(J_{r}%
^{n})_{q_{l}}A_{s}^{lj}%
\]
and thus
\begin{align*}
&  -\frac{1}{2}m\left(  [J_{r}^{l}\partial_{l},q_{n}^{-1}A_{s}^{nj}%
\partial_{j}]+[q_{n}^{-1}A_{r}^{nj}\partial_{j},J_{s}^{l}\partial_{l}]\right)
\\
&  =-\frac{1}{2}mq_{n}^{-1}\left(  J_{r}^{l}(A_{s}^{nj})_{q_{l}}-A_{s}%
^{nl}(J_{r}^{j})_{q_{l}}+A_{r}^{nl}(J_{s}^{j})_{q_{l}}-J_{s}^{l}(A_{r}%
^{nj})_{q_{l}}-A_{s}^{lj}(J_{r}^{n})_{q_{l}}+A_{r}^{lj}(J_{s}^{n})_{q_{l}%
}\right)  \partial_{j}\\
&  \overset{(\ref{q10})}{=}\frac{1}{2}m(s-r)q_{n}^{-1}A_{r+s-(n-m+2)}%
^{nj}\partial_{j}.
\end{align*}
Gathering all the terms in (\ref{zbieraj}) as calculated above and comparing
the result with (\ref{q8}) we receive the relation (\ref{q13}). This concludes
the proof of Lemma \ref{l3}.
\end{proof}

Lemmas \ref{l2} and \ref{l3} immediately imply the thesis of Theorem
\ref{algebra} which in turn implies that the relation (\ref{wk1}) is true and
therefore Theorem \ref{MAIN} is proved.

The first example (Example \ref{PRZ1} below) is rather detailed, to illustrate
various aspects of the theory presented above.

\begin{example}
\label{PRZ1}(Non-autonomous quantum H\'{e}non-Heiles system. Continuation of
Example 3 from Part I) Choose $n=2$ and $m=1$. Then $I_{1}^{m}=\left\{
2\right\}  $ while $I_{2}^{m}=\emptyset$ and the only Killing vector field
(\ref{Zd}) is $J_{2}=\partial_{1}=\frac{\partial}{\partial q_{1}}$ (and there
are no fields (\ref{Zg})). The curve (\ref{sA}) becomes%
\[
\sum_{\alpha=-1}^{5}c_{\alpha}(t)x^{\alpha}+h_{1}x+h_{2}=\frac{1}{2}%
xy^{2}\quad
\]
and solving the corresponding separation relations yields the Hamiltonians
$h_{r}$ that in Vi\`{e}te coordinates $(q,p)$ attain the form%
\[
h_{r}=\frac{1}{2}\,p^{T}A_{r}{p}+\sum_{\alpha=-1}^{5}c_{\alpha}(t)V_{r}%
^{(\alpha)}\text{, \ }r=1,2
\]
with $A_{r}=K_{r}G$ of the form
\begin{equation}
A_{1}=\left(
\begin{array}
[c]{cc}%
1 & 0\\
0 & -q_{2}%
\end{array}
\right)  \text{, \ \ }A_{2}=\left(
\begin{array}
[c]{cc}%
0 & -q_{2}\\
-q_{2} & -q_{1}q_{2}%
\end{array}
\right)  \label{Arp1}%
\end{equation}
and where $V^{(0)}=(0,-1)^{T}$ and $V^{(1)}=(-1,0)$ are trivial potentials and
where%
\begin{align*}
V^{(-1)}  &  =(q_{2}^{-1},q_{1}q_{2}^{-1})^{T}\text{, \ }V^{(2)}=(q_{1}%
,q_{2})^{T}\text{, \ }V^{(3)}=(q_{2}-q_{1}^{2},-q_{1}q_{2})^{T}\text{,
\ \ }V^{(4)}=(q_{1}^{3}-2q_{1}q_{2},q_{1}^{2}q_{2}-q_{2}^{2})^{T}\\
V^{(5)}  &  =(-q_{1}^{4}+3q_{1}^{2}q_{2}-q_{2}^{2},-q_{1}^{3}q_{2}+2q_{1}%
q_{2}^{2})^{T}%
\end{align*}
Deforming $h_{1}$ and $h_{2}$ respectively by the functions $W_{1}=0$ and
$W_{2}=p_{1}$ yields the Hamiltonians $h_{1}^{A}=h_{1}$ and $h_{2}^{A}%
=h_{2}+p_{1}$. Further, due to the fact that $\kappa_{1}=2=n$ we have that
$H_{r}^{A}=h_{r}^{A}$ for $r=1,2$ (see (\ref{kappy}) and (\ref{7b})). We can
now determine the functions $c_{\alpha}(t)$ by demanding that $H_{r}^{A}$
satisfy the Frobenius condition (\ref{fcg}) with $f_{12}\equiv0$. That leads
to an overdetermined, but solvable, system of PDE's for $c_{\alpha}$ (note
that the functions $c_{\alpha}$ do not depend on the number $r$ of the
Hamiltonian) with a particular solution
\[
c_{-1}=-\frac{1}{4}a\text{ arbitrary, \ }c_{0}=\frac{1}{2}t_{1}^{2}%
+3t_{1}t_{2}^{2},\text{ \ }c_{1}=0\text{, \ }c_{2}=-(t_{1}+3t_{2}^{2})\text{,
}c_{3}=-3t_{2}\text{, }c_{4}=-1\text{, }c_{5}=0.
\]
Another particular solution, which differs from the one above only in the
non-dynamical (in the sense that they do not influence the Hamiltonian flows
of $h_{r}^{A}$) functions $c_{0}$ and $c_{1}$ is%
\[
c_{-1}=-\frac{1}{4}a\text{ arbitrary, \ }c_{0}=\frac{1}{2}t_{1}^{2},\text{
\ }c_{1}=-t_{2}^{3}\text{, \ }c_{2}=-(t_{1}+3t_{2}^{2})\text{, }c_{3}%
=-3t_{2}\text{, }c_{4}=-1\text{, }c_{5}=0.
\]
In consequence, the sought Hamiltonians $H_{r}^{A}$ become
\begin{align*}
H_{1}^{A}  &  =\frac{1}{2}\,{p}_{1}^{2}-\frac{1}{2}\,q_{{2}}{p}_{2}^{2}%
-q_{1}^{3}+2q_{1}q_{2}+3t_{2}(q_{1}^{2}-q_{2})-(t_{1}+3t_{2}^{2})q_{1}%
-\frac{1}{4}a\,q_{2}^{-1}-c_{1},\\
H_{2}^{A}  &  =-q_{{2}}p_{{2}}p_{{1}}-\frac{1}{2}\,q_{{1}}q_{{2}}{p}_{2}%
^{2}+p_{1}+q_{2}^{2}-q_{1}^{2}q_{2}+3t_{2}q_{1}q_{2}-(t_{1}+3t_{2}^{2}%
)q_{2}-\frac{1}{4}aq_{1}q_{2}^{-1}-c_{0},
\end{align*}
and one can verify, by direct computation, that they satisfy the Frobenius
condition (\ref{fcg}) with $f_{rs}=0$. In the flat orthogonal coordinates
$(x_{1},x_{2},y_{1},y_{2})$ \cite{blasz2007} (the reader should not confuse
these coordinates with the variables $x,y$ used in the considered algebraic
curves)
\[
q_{1}=-x_{1},\quad q_{2}=-\frac{1}{4}x_{2}^{2},\quad p_{1}=-y_{1},\quad
p_{2}=-\frac{2y_{2}}{x_{2}},
\]
the Hamiltonians $H_{i}^{A}$ take the form
\begin{align*}
H_{1}^{A}  &  =\frac{1}{2}y_{1}^{2}+\frac{1}{2}y_{2}^{2}+x_{1}^{3}+\frac{1}%
{2}x_{1}x_{2}^{2}+a\,x_{2}^{-2}+3t_{2}(x_{1}^{2}+\frac{1}{4}x_{2}^{2}%
)+(t_{1}+3t_{2}^{2})x_{1}-c_{1}\\
&  \equiv h_{1}^{HH}+3t_{2}(x_{1}^{2}+\frac{1}{4}x_{2}^{2})+(t_{1}+3t_{2}%
^{2})x_{1}-c_{1},\\
H_{2}^{A}  &  =\frac{1}{2}x_{2}y_{1}y_{2}-\frac{1}{2}x_{1}y_{2}^{2}%
-y_{1}+\frac{1}{16}x_{2}^{4}+\frac{1}{4}x_{1}^{2}x_{2}^{2}-a\,x_{1}x_{2}%
^{-2}+\frac{1}{4}3t_{2}x_{1}x_{2}^{2}+\frac{1}{4}(t_{1}+3t_{2}^{2})x_{2}%
^{2}-c_{0}\\
&  \equiv h_{2}^{HH}+\frac{1}{4}3t_{2}x_{1}x_{2}^{2}+\frac{1}{4}(t_{1}%
+3t_{2}^{2})x_{2}^{2}-c_{0}%
\end{align*}
(so that $G=I$ and the coordinates are actually Euclidean) and constitute a
non-autonomous deformation of the integrable case of the extended
H\'{e}non-Heiles system $h_{r}^{HH}$. Moreover, the flow generated by
$h_{1}^{HH}$ is exactly the stationary flow of the $5th$-order KdV
\cite{Allan}. Let us now perform the minimal quantization of the obtained
$H_{i}^{A}$ in the flat coordinates. The $(2,0)$-tensors $A_{r}$ (\ref{Arp1})
have the form
\[
A_{1}=\operatorname{Id}\text{, }A_{2}=\left(
\begin{array}
[c]{cc}%
0 & \frac{1}{2}x_{2}\\
\frac{1}{2}x_{2} & -x_{1}%
\end{array}
\right)
\]
and the minimal quantization $\widehat{H}_{r}$ of $H_{r}$ can be calculated
using (\ref{q2}) and (\ref{q3}). The result is%
\begin{align*}
\widehat{H}_{1}^{A}  &  =-\frac{1}{2}\hslash^{2}\left(  \partial_{1}%
^{2}+\partial_{2}^{2}\right)  +x_{1}^{3}+\frac{1}{2}x_{1}x_{2}^{2}%
+\alpha\,x_{2}^{-2}+3t_{2}(x_{1}^{2}+\frac{1}{4}x_{2}^{2})+(t_{1}+3t_{2}%
^{2})x_{1}-c_{1},\\
\widehat{H}_{2}^{A}  &  =-\frac{1}{2}\hslash^{2}\left(  x_{2}\partial
_{1}\partial_{2}-x_{1}\partial_{2}^{2}+\frac{1}{2}\partial_{1}\right)
+i\hslash\partial_{1}+\frac{1}{16}x_{2}^{4}+\frac{1}{4}x_{1}^{2}x_{2}%
^{2}-\alpha\,x_{1}x_{2}^{-2}+\frac{1}{4}3t_{2}x_{1}x_{2}^{2}+\frac{1}{4}%
(t_{1}+3t_{2}^{2})x_{2}^{2}-c_{0},
\end{align*}
where $\partial_{i}=\frac{\partial}{\partial x_{i}}$ now, and one can show by
a direct computation that $\widehat{H}_{r}^{A}$ do indeed satisfy the quantum
Frobenius condition (\ref{fcq}).
\end{example}

\begin{example}
(Example 2 from Part I continued) Let us now choose $n=3$ and $m=1$. This time
we will perform calculations in Vi\`{e}te coordinates $(p,q)$. We have
$I_{1}^{m}=\left\{  2,3\right\}  $ and $I_{2}^{m}=\emptyset$, while the vector
fields $J_{r}$ are
\[
\text{\ }J_{2}=\partial_{2}\text{, \ }J_{3}=2\partial_{1}+q_{1}\partial_{2}%
\]
so that $W_{1}=0$, $W_{2}=p_{2}$ while $W_{3}=2p_{1}+q_{1}p_{2}$. The curve
(\ref{sA}) becomes%
\[
\sum_{\alpha=-1}^{7}c_{\alpha}(t)x^{\alpha}+h_{1}x^{2}+h_{2}x+h_{3}=\frac
{1}{2}xy^{2}\quad
\]
and the Hamiltonians $h_{r}$ in become%
\[
h_{r}=\frac{1}{2}\mu^{T}A_{r}\mu+\sum_{\alpha=-1}^{7}c_{\alpha}(t)V_{r}%
^{(\alpha)},\quad r=1,\dotsc,3,
\]
with $A_{r}$ in Vi\`{e}te coordinates given by (see (\ref{n1}) and
(\ref{n2})):%
\[
A_{1}=G=\left(
\begin{array}
[c]{ccc}%
0 & 1 & 0\\
1 & q_{1} & 0\\
0 & 0 & -q_{3}%
\end{array}
\right)  \text{, \ }A_{2}=\left(
\begin{array}
[c]{ccc}%
1 & q_{1} & 0\\
q_{1} & q_{1}^{2}-q_{2} & -q_{3}\\
0 & -q_{3} & -q_{1}q_{3}%
\end{array}
\right)  \text{, \ \ }A_{3}=\left(
\begin{array}
[c]{ccc}%
0 & 0 & -q_{3}\\
0 & -q_{3} & -q_{1}q_{3}\\
-q_{3} & -q_{1}q_{3} & -q_{2}q_{3}%
\end{array}
\right)
\]
and where the scalar potentials $V_{r}^{(\alpha)}$ are polynomials in $q$ that
can be calculated from the recursion formula in Part I. Further, $h_{r}%
^{A}=h_{r}+W_{r}$ and, since $\kappa_{1}=2<n$ and due to (\ref{7b}) and to
results in Part I
\begin{equation}
H_{1}^{A}=h_{1}^{A}\text{, \ }H_{2}^{A}=h_{2}^{A}\text{, \ }H_{3}^{A}%
=h_{3}^{A}+t_{2}h_{1}^{A}. \label{hp2}%
\end{equation}
The functions $c_{\alpha}$ can now be determined from the Frobenius condition
(\ref{fcg}) (with $f_{rs}=0$). It leads to an overdetermined but soluble
system of PDE's. A particular solution of this system is%
\begin{align}
c_{-1}  &  =c_{0}=0\text{, \ }c_{1}=a(t_{3}^{3}-2t_{2}t_{3}-4t_{1}%
)t_{3}\text{, \ }c_{2}=4a(t_{3}^{3}-t_{1})\text{, }\label{hp2stale}\\
c_{3}  &  =2a(3t_{3}^{2}+t_{2})\text{, }c_{4}=4at_{3}\text{, }c_{5}%
=c_{6}=c_{7}=0\text{, \ }a\in\mathbb{R}\text{.}\nonumber
\end{align}
Let us now perform the minimal quantization of the Hamiltonians $H_{r}$
(\ref{hp2}) with $c_{\alpha}$ given by (\ref{hp2stale}). Of course%
\begin{equation}
\widehat{H}_{1}^{A}=\widehat{h}_{1}^{A}\text{, \ }\widehat{H}_{2}%
^{A}=\widehat{h}_{2}^{A}\text{, \ }\widehat{H}_{3}^{A}=\widehat{h}_{3}%
^{A}+t_{2}\widehat{h}_{1}^{A} \label{hp2q}%
\end{equation}
where $\widehat{h}_{i}^{A}$ can be computed explicitly from (\ref{q8}). The
final result is%
\[
\widehat{h}_{r}^{A}=\widehat{E}_{r}-i\hslash J_{r}+a(t_{3}^{3}-2t_{2}%
t_{3}-4t_{1})t_{3}V_{r}^{(1)}+4a(t_{3}^{3}-t_{1})V_{r}^{(2)}+2a(3t_{3}%
^{2}+t_{2})V_{r}^{(3)}+4at_{3}V_{r}^{(4)}\text{, }r=1,2,3
\]
where%
\begin{align*}
\widehat{E}_{1}  &  =-\frac{1}{2}\hslash^{2}\left(  2\partial_{1}\partial
_{2}+q_{1}\partial_{2}^{2}-q_{3}\partial_{3}^{2}-\frac{1}{2}\partial
_{3}\right) \\
\widehat{E}_{2}  &  =-\frac{1}{2}\hslash^{2}\left(  \partial_{1}^{2}%
+2q_{1}\partial_{1}\partial_{2}+(q_{1}^{2}-q_{2})\partial_{2}^{2}%
-2q_{3}\partial_{2}\partial_{3}-q_{1}q_{3}\partial_{3}^{2}-\frac{1}{2}%
\partial_{2}-\frac{1}{2}q_{1}\partial_{3}\right) \\
\widehat{E}_{3}  &  =-\frac{1}{2}\hslash^{2}\left(  -2q_{3}\partial
_{1}\partial_{3}-q_{3}\partial_{2}^{2}-2q_{1}q_{3}\partial_{2}\partial
_{3}-q_{2}q_{3}\partial_{3}^{2}-\frac{1}{2}\partial_{1}-\frac{1}{2}%
q_{1}\partial_{2}-\frac{1}{2}q_{2}\partial_{3}\right)
\end{align*}
and the resulting quantum operators $\widehat{H}_{i}^{A}$ in (\ref{hp2q}) do
satisfy the quantum Frobenius condition (\ref{fcq}).
\end{example}

\subsection{Quantization of systems with magnetic terms}

We proceed now with the case of magnetic potentials. According to Definition
\ref{minq}, the minimal quantization $\widehat{H}_{r}^{B}$\ of Hamiltonians
$H_{r}^{B}$ is in the magnetic case obtained by replacing the Hamiltonians
$h_{r}^{A}$ in (\ref{7b}), (\ref{7c}) (with $h_{r}$ given now by the magnetic
Hamiltonians (\ref{mag1})) by the self-adjoint operators
\begin{equation}
\widehat{h}_{r}^{B}=\widehat{E}_{r}+\widehat{W}_{r}+\sum_{\gamma=0}%
^{n+1}d_{\gamma}(t)\widehat{M}_{r}^{(\gamma)},\ \ \ \ \ r=1,...,n, \label{q1m}%
\end{equation}
where $\widehat{E}_{r}$ and $\widehat{W}_{r}$ are given, as before by
(\ref{q8}), while the quantized magnetic terms $\widehat{M}_{r}^{(\gamma)}$
are given by
\begin{align}
\widehat{M}_{r}^{(\gamma)}  &  =-\frac{1}{2}i\hslash\left[  \nabla_{j}%
(P_{r}^{(\gamma)})^{j}+(P_{r}^{(\gamma)})^{j}\nabla_{j}\right]  =-\frac{1}%
{2}i\hslash\left[  \left\vert g\right\vert ^{-\frac{1}{2}}\partial
_{j}\left\vert g\right\vert ^{\frac{1}{2}}(P_{r}^{(\gamma)})^{j}%
+(P_{r}^{(\gamma)})^{j}\partial_{j}\right] \label{pm}\\
&  =-\frac{1}{2}i\hslash(\gamma-\frac{1}{2}m)V_{r}^{(\gamma-1)}-i\hslash
(P_{r}^{(\gamma)})^{j}\partial_{j},\text{ \ \ }\partial_{j}=\frac{\partial
}{\partial q_{j}}\nonumber
\end{align}
where the explicit form of the first-order operator $\widehat{M}_{r}%
^{(\gamma)}$ in (\ref{pm}) is obtained by straightforward calculations. Note
that the potentials $V_{r}^{(\gamma-1)}$ in (\ref{pm}) are nontrivial only for
$\gamma=0$ and $\gamma=n+1$. Note also that the minimal quantization
$\widehat{M}_{r}^{(\gamma)}$ of the classical magnetic term (\ref{MT}) depends
on $m$ even though $M_{r}^{(\gamma)}=p^{T}P_{r}^{(\gamma)}$ itself does not
depend on $m$. As in the non-magnetic case, we have

\begin{theorem}
\label{MAINM}For any fixed $m\in\left\{  0,\ldots,n+1\right\}  $, the set of
$n$ operators $\widehat{H}_{r}^{B}$, $r=1,\ldots n$, given by (\ref{7b}) and
(\ref{7c}) with $h_{r}^{A}$ replaced by $\widehat{h}_{r}^{B}$ given by
(\ref{q1m}), satisfies the quantum Frobenius condition (\ref{fcq}).
\end{theorem}

By the same reasons as in the non-magnetic case, the Hamiltonians
$\widehat{H}_{r}^{B}$ will satisfy the quantum Frobenius condition (\ref{fcq})
as soon as any pair of operators $\widehat{\mathcal{E}}_{r}+\widehat{M}%
_{r}^{(\gamma)}\equiv\widehat{E}_{r}+\widehat{W}_{r}+\widehat{M}_{r}%
^{(\gamma)}$ and $\widehat{\mathcal{E}}_{s}+\widehat{M}_{s}^{(\gamma)}%
\equiv\widehat{E}_{s}+\widehat{W}_{s}+\widehat{M}_{s}^{(\gamma)}$ will satisfy
the same, up to the factor $i\hslash$, commutation relations as their
classical counterparts $\mathcal{E}_{r}+M_{r}^{(\gamma)}\equiv E_{r}%
+W_{r}+M_{r}^{(\gamma)}$ and $\mathcal{E}_{s}+M_{s}^{(\gamma)}\equiv
E_{s}+W_{s}+M_{s}^{(\gamma)}$ (again, we perform all calculations in Vi\`{e}te
coordinates). Thus, we have to show that for any $r,s\in\left\{
1,\ldots,n\right\}  $, any $m\in\left\{  0,\ldots,n+1\right\}  $ and any
$\gamma\in\left\{  0,\ldots,n+1\right\}  $\qquad%
\begin{equation}
i\hslash\widehat{\left\{  E_{r}+W_{r}+M_{r}^{(\gamma)},E_{s}+W_{s}%
+M_{s}^{(\gamma)}\right\}  }=\left[  \widehat{E}_{r}+\widehat{W}%
_{r}+\widehat{M}_{r}^{(\gamma)},\widehat{E}_{s}+\widehat{W}_{s}+\widehat{M}%
_{s}^{(\gamma)}\right]  \label{wk2}%
\end{equation}
as the relation (\ref{wk2}) means that we can perform the appropriate
deformation procedure described in Part I (and shortly revisited in Section
\ref{sec 2}) also in the quantum magnetic case and that this procedure will
yield the same PDE's for functions $d_{\gamma}(t_{1},\ldots,t_{n})$,
$\zeta_{r,j}(t_{1},\dotsc,t_{r-1})$ and $\zeta_{r,r+j}(t_{r+1},\dotsc,t_{n})$
as the corresponding classical procedure.

For any $r,s\in\left\{  1,\ldots,n\right\}  $%

\begin{align*}
\left\{  E_{r}+W_{r}+M_{r}^{(\gamma)},E_{s}+W_{s}+M_{s}^{(\gamma)}\right\}
&  =\left\{  \mathcal{E}_{r},\mathcal{E}_{s}\right\}  +\left\{  E_{r}%
,M_{s}^{(\gamma)}\right\}  +\left\{  M_{r}^{(\gamma)},E_{s}\right\} \\
&  +\left\{  W_{r},M_{s}^{(\gamma)}\right\}  +\left\{  M_{r}^{(\gamma)}%
,W_{s}\right\}  +\left\{  M_{r}^{(\gamma)},M_{s}^{(\gamma)}\right\}  .
\end{align*}
Moreover, the magnetic St\"{a}ckel Hamiltonians themselves (i.e. without the
quasi-St\"{a}ckel term $W_{r}$), both geodesic and with magnetic terms,
commute with each other, so that $\left\{  E_{r},E_{s}\right\}  =0$ and
$\left\{  E_{r}+M_{r}^{(\gamma)},E_{s}+M_{s}^{(\gamma)}\right\}  =0$, which
yields%
\[
\left\{  E_{r},M_{s}^{(\gamma)}\right\}  +\left\{  M_{r}^{(\gamma)}%
,E_{s}\right\}  +\left\{  M_{r}^{(\gamma)},M_{s}^{(\gamma)}\right\}  =0
\]
and thus the left hand side of (\ref{wk2}) is $\left\{  \mathcal{E}%
_{r},\mathcal{E}_{s}\right\}  ,$ is
\begin{equation}
i\hslash\widehat{\left\{  \mathcal{E}_{r},\mathcal{E}_{s}\right\}  }%
+i\hslash\left(  \widehat{\left\{  W_{r},M_{s}^{(\gamma)}\right\}  +\left\{
M_{r}^{(\gamma)},W_{s}\right\}  }\right)  . \label{klasm}%
\end{equation}
Moreover, due to (\ref{2.6}),
\[
\left[  \widehat{E}_{r},\widehat{E}_{s}\right]  =0\text{ and }\left[
\widehat{E}_{r}+\widehat{M}_{r}^{(\gamma)},\widehat{E}_{s}+\widehat{M}%
_{s}^{(\gamma)}\right]  =0
\]
and thus
\[
\left[  \widehat{E}_{r},\widehat{M}_{s}^{(\gamma)}\right]  +\left[
\widehat{M}_{r}^{(\gamma)},\widehat{E}_{s}\right]  +\left[  \widehat{M}%
_{r}^{(\gamma)},\widehat{M}_{s}^{(\gamma)}\right]  =0
\]
so that the right hand side of (\ref{wk2}) is actually%

\[
\lbrack\widehat{\mathcal{E}}_{r},\widehat{\mathcal{E}}_{s}]+\left[
\widehat{E}_{r},\widehat{M}_{s}^{(\gamma)}\right]  +\left[  \widehat{M}%
_{r}^{(\gamma)},\widehat{E}_{s}\right]
\]
By Theorem \ref{algebra}, $i\hslash\widehat{\left\{  \mathcal{E}%
_{r},\mathcal{E}_{s}\right\}  }=[\widehat{\mathcal{E}}_{r}%
,\widehat{\mathcal{E}}_{s}]$ so in order to prove (\ref{wk2}) and thus Theorem
\ref{MAINM} it remains to show that the following relation is valid:
\begin{equation}
i\hslash\widehat{\left(  \left\{  W_{r},M_{s}^{(\gamma)}\right\}  +\left\{
M_{r}^{(\gamma)},W_{s}\right\}  \right)  }=\left[  \widehat{W}_{r}%
,\widehat{M}_{s}^{(\gamma)}\right]  +\left[  \widehat{M}_{r}^{(\gamma
)},\widehat{W}_{s}\right]  . \label{toprove}%
\end{equation}

A direct calculation shows that the right hand side of (\ref{toprove}) is%
\begin{equation}
-\hslash^{2}\left(  \left[  J_{r},P_{s}^{(\gamma)}\right]  +\left[
P_{r}^{(\gamma)},J_{s}\right]  +aJ_{r}\left(  V_{s}^{(\gamma-1)}\right)
-aJ_{s}\left(  V_{r}^{(\gamma-1)}\right)  \right)  \label{P}%
\end{equation}
with $a$ denoting (in this proof) $\frac{1}{2}(\gamma-\frac{1}{2}m)$, so that
it is a sum of a vector field and a function on $\mathcal{M}$. Further%
\[
\left\{  W_{r},M_{s}^{(\gamma)}\right\}  +\left\{  M_{r}^{(\gamma)}%
,W_{s}\right\}  =\Theta^{i}p_{i}%
\]
for some vector field $\Xi=\Xi^{i}\frac{\partial}{\partial q_{i}}$ on
$\mathcal{M}$. Thus, due to (\ref{qq3}), the left hand side of (\ref{toprove})
is%
\begin{equation}
\hslash^{2}\left(  \Theta+\tfrac{1}{2}\left(  \partial_{j}\Theta^{j}\right)
+\tfrac{1}{4}\left\vert g\right\vert ^{-1}\left(  \partial_{j}\left\vert
g\right\vert \right)  \Theta^{j}\right)  \label{L}%
\end{equation}
and as such is also a sum of a vector field and a function on $\mathcal{M}$.
The vector field parts of (\ref{P}) and (\ref{L}) are equal due to (\ref{S})
so it remains to prove that%
\begin{equation}
\tfrac{1}{2}\left(  \partial_{j}\Theta^{j}\right)  +\tfrac{1}{4}\left\vert
g\right\vert ^{-1}\left(  \partial_{j}\left\vert g\right\vert \right)
\Theta^{j}=aJ_{r}\left(  V_{s}^{(\gamma-1)}\right)  -aJ_{s}\left(
V_{r}^{(\gamma-1)}\right)  \label{LP}%
\end{equation}
A direct calculation shows that both sides of (\ref{LP}) are zero for all
$\gamma$ except $\gamma=0$ and $\gamma=n+1$. For $\gamma=n+1$ both sides of
(\ref{LP}) for a given choice of indices $(r,s)$ are%
\[%
\begin{cases}
0, & \text{for }r\in I_{1}^{m}\text{ and }s\in I_{2}^{m},\\
a(r-s)q_{r+s-(n-m+2)}, & \text{for }r,s\in I_{1}^{m},\\
-a(r-s)q_{r+s-(n-m+2)},, & \text{for }r,s\in I_{2}^{m}.
\end{cases}
\]
while for $\gamma=0$ both sides of (\ref{LP}) for a given choice of indices
$(r,s)$ are%
\[
\frac{a(m-1)}{q_{n}}\left(  q_{r-1}\delta_{s,n-m+2}-q_{s-1}\delta
_{r,n-m+2}\right)  +%
\begin{cases}
0, & \text{for }r\in I_{1}^{m}\text{ and }s\in I_{2}^{m},\\
a(r-s)q_{r+s-(n-m+2)}, & \text{for }r,s\in I_{1}^{m},\\
-a(r-s)q_{r+s-(n-m+2)}, & \text{for }r,s\in I_{2}^{m}.
\end{cases}
.
\]

Therefore, (\ref{toprove}) is valid. This also concludes the proof of Theorem
\ref{MAINM}.

\begin{example}
\label{4e}Consider the case $n=3$ and $m=3$ (see subsection 7.2 in Part I),
which means that $I_{1}^{m}=\emptyset$ while $I_{2}^{m}=\left\{  2,3\right\}
$. Further, $\kappa_{1}=\kappa_{2}=1$ and thus the vector fields $J_{r}$ are
\[
J_{2}=q_{2}\partial_{2}+2q_{3}\partial_{3}\text{, \ }J_{3}=q_{3}\partial_{2}%
\]
so that $W_{1}=0$, $W_{2}=p_{2}q_{2}+2q_{3}p_{3}$ while $W_{3}=q_{3}p_{2}$.
The curve (\ref{mag}) becomes%
\[
\sum_{\gamma=0}^{4}d_{\gamma}(t)x^{\gamma}y+x^{2}h_{1}+xh_{2}+h_{3}=\frac
{1}{2}x^{3}y^{2}%
\]
leading to Hamiltonians $h_{r}$ (\ref{mag1}) that in Vi\`{e}te coordinates
attain the form%
\[
h_{r}=\frac{1}{2}p^{T}A_{r}p+\sum_{\gamma=0}^{4}d_{\alpha}(t)p^{T}%
P_{r}^{(\gamma)},\quad r=1,\dotsc,3,
\]
with $A_{r}$ given by (see (\ref{n1}) and (\ref{n2})):
\[
A_{1}=G=\left(
\begin{array}
[c]{ccc}%
-q_{1} & -q_{2} & -q_{3}\\
-q_{2} & -q_{3} & 0\\
-q_{3} & 0 & 0
\end{array}
\right)  \text{, \ }A_{2}=\left(
\begin{array}
[c]{ccc}%
-q_{2} & -q_{3} & 0\\
-q_{3} & -q_{1}q_{3}+q_{2}^{2} & q_{2}q_{3}\\
0 & q_{2}q_{3} & q_{3}^{2}%
\end{array}
\right)  \text{, \ \ }A_{3}=\left(
\begin{array}
[c]{ccc}%
-q_{3} & 0 & 0\\
0 & q_{2}q_{3} & q_{3}^{2}\\
0 & q_{3}^{2} & 0
\end{array}
\right)  .
\]
Further, the vector potentials $P_{r}^{(\gamma)}$ are given by (\ref{Pwq}).
Explicitly
\begin{align}
P_{1}^{(0)}  &  =\text{ }(0,0,1)^{T}\text{, \ }P_{2}^{(0)}=(0,1,q_{1}%
)^{T}\text{, \ }P_{3}^{(0)}=(1,q_{1},q_{2})^{T},\nonumber\\
P_{1}^{(1)}  &  =\text{ }(0,1,0)^{T}\text{, \ }P_{2}^{(1)}=(1,q_{1}%
,0)^{T}\text{, \ }P_{3}^{(1)}=(0,0,-q_{3})^{T},\nonumber\\
P_{1}^{(2)}  &  =\text{ }(1,0,0)^{T}\text{, \ }P_{2}^{(2)}=(0,-q_{2}%
,-q_{3})^{T}\text{, \ }P_{3}^{(2)}=(0,-q_{3},0)^{T},\label{Pr}\\
P_{1}^{(3)}  &  =\text{ }(-q_{1},-q_{2},-q_{3})^{T}\text{, \ }P_{2}%
^{(3)}=(-q_{2},-q_{3},0)^{T}\text{, \ }P_{3}^{(3)}=(-q_{3},0,0)^{T},\text{
}\nonumber\\
P_{1}^{(4)}  &  =\text{ }(q_{1}^{2}-q_{2},q_{1}q_{2}-q_{3},q_{1}q_{3}%
)^{T}\text{, \ }P_{2}^{(4)}=(q_{1}q_{2}-q_{3},q_{2}^{2},q_{2}q_{3})^{T}\text{,
\ }P_{3}^{(4)}=(q_{1}q_{3},q_{2}q_{3},q_{3}^{2})^{T}.\nonumber
\end{align}
As usual, $h_{r}^{B}=h_{r}+W_{r}$ and, since $\kappa_{1}=\kappa_{2}<n$ and due
to (\ref{7b}), (\ref{7c}) and results in Part I
\begin{equation}
H_{1}^{B}=h_{1}^{B}\text{, \ }H_{2}^{B}=h_{2}^{B}+t_{3}h_{3}^{B}\text{,
\ }H_{3}^{B}=h_{3}^{B}. \label{hp3q}%
\end{equation}
The functions $d_{\alpha}$ can now be determined from the Frobenius condition
(\ref{fcg}) (with $f_{rs}=0$). It leads to an overdetermined but soluble
system of PDE's. The general solution of this system is%
\[
d_{0}=b_{0}\exp(2t_{2})\text{, \ }d_{1}=b_{0}t_{3}\exp(2t_{2})+b_{1}\exp
(t_{2})\text{, \ }d_{2}=b_{2}\text{, }d_{3}=(b_{3}+b_{4}t_{1})\text{, \ }%
d_{4}=b_{4},
\]
(parametrized by the arbitrary constants $b_{r}$). Let us now minimally
quantize the Hamiltonians $H_{r}$ in (\ref{hp3q}). Naturally%
\[
\widehat{H}_{1}^{B}=\widehat{h}_{1}^{B}\text{, \ }\widehat{H}_{2}%
^{B}=\widehat{h}_{2}^{B}+t_{3}\widehat{h}_{3}^{B}\text{, \ }\widehat{H}%
_{3}^{B}=\widehat{h}_{3}^{B}%
\]
where $\widehat{h}_{r}^{B}$ are given by (\ref{q1m}) where $\widehat{E}_{r}$
and $\widehat{W}_{r}$ can computed using (\ref{q8}) and where $\widehat{M}%
_{r}^{(\gamma)}\,$\ can be computed by (\ref{pm}). The result is%
\[
\widehat{h}_{r}^{B}=\widehat{E}_{r}-i\hslash J_{r}+b_{0}\exp(2t_{2}%
)\widehat{M}_{r}^{(0)}+[b_{0}t_{3}\exp(2t_{2})+b_{1}\exp(t_{2})]\widehat{M}%
_{r}^{(1)}+b_{2}\widehat{M}_{r}^{(2)}+(b_{3}+b_{4}t_{1})\widehat{M}_{r}%
^{(3)}+b_{4}\widehat{M}_{r}^{(4)}\text{,}%
\]
where%
\begin{align*}
\widehat{E}_{1}  &  =-\frac{1}{2}\hslash^{2}\left(  -q_{1}\partial_{1}%
^{2}-q_{3}\partial_{2}^{2}-2q_{2}\partial_{1}\partial_{2}-2q_{3}\partial
_{1}\partial_{3}-\frac{3}{2}\partial_{1}\right)  ,\\
\widehat{E}_{2}  &  =-\frac{1}{2}\hslash^{2}\left(  -q_{2}\partial_{1}%
^{2}+(-q_{1}q_{3}+q_{2}^{2})\partial_{2}^{2}+q_{3}^{2}\partial_{3}^{2}%
-2q_{3}\partial_{1}\partial_{2}+2q_{2}q_{3}\partial_{2}\partial_{3}+\frac
{3}{2}q_{2}\partial_{2}+\frac{3}{2}q_{3}\partial_{3}\right)  ,\\
\widehat{E}_{3}  &  =-\frac{1}{2}\hslash^{2}\left(  -q_{3}\partial_{1}%
^{2}+q_{2}q_{3}\partial_{2}^{2}-2q_{3}^{2}\partial_{2}\partial_{3}+\frac{3}%
{2}q_{3}\partial_{2}\right)
\end{align*}
and where%
\begin{align*}
\widehat{M}_{r}^{(0)}  &  =\frac{3}{4}i\hslash\frac{q_{r-1}}{q_{n}}%
-i\hslash(P_{r}^{(0)})^{j}\partial_{j},\\
\widehat{M}_{r}^{(1)}  &  =-\frac{1}{4}i\hslash\delta_{r,3}-i\hslash
(P_{r}^{(1)})^{j}\partial_{j},\\
\widehat{M}_{r}^{(2)}  &  =\frac{1}{4}i\hslash\delta_{r,2}-i\hslash
(P_{r}^{(2)})^{j}\partial_{j},\\
\widehat{M}_{r}^{(3)}  &  =\frac{3}{4}i\hslash\delta_{r,1}-i\hslash
(P_{r}^{(3)})^{j}\partial_{j},\\
\widehat{M}_{r}^{(4)}  &  =-\frac{5}{4}i\hslash q_{r}-i\hslash(P_{r}%
^{(4)})^{j}\partial_{j}%
\end{align*}
with $r=1,2,3$ and with $P_{r}^{(\gamma)}$ given by (\ref{Pr}). It can be
demonstrated by a direct computation that $\widehat{H}_{r}^{B}$ do satisfy the
quantum Frobenius condition (\ref{fcq}).
\end{example}

\section{Quantum canonical transformations between magnetic and non-magnetic
quantum Painlev\'{e} systems\label{s5}}

In this chapter we prove that the magnetic quantum Painlev\'{e} operator
$\,\widehat{H}_{r}^{B}$ can be transformed, by a multitime-dependent quantum
canonical transformations (see \cite{kim}), to a corresponding non-magnetic
quantum Painlev\'{e} \ operator $\widehat{H}_{r}^{A}$. Since $\widehat{H}%
_{r}^{A}$ contains $n+3$ parameters while $\,\widehat{H}_{r}^{B}$ contains
only $n+2$ parameters, we have first to extend the quantum magnetic system
$\widehat{h}_{r}^{B}$ by one parameter to the system%
\begin{equation}
\widehat{h}_{r}^{B}=\widehat{E}_{r}+\widehat{W}_{r}+\sum_{\gamma=0}%
^{n+1}d_{\gamma}(t)\widehat{M}_{r}^{(\gamma)}+\overline{b}e_{n}(t)V_{r}%
^{(n)}-\overline{b}e_{n-r}(t),\text{ \ \ }r=1,\ldots,n \label{q1me}%
\end{equation}
(as we also did in the classical case, see Section 8 of Part I) where from the
Frobenius condition it follows that $e_{n}(t_{1},\dotsc,t_{n})=1$ for
$m=0,\dotsc,n$ and $e_{n}(t_{1},\dotsc,t_{n})=\exp(t_{1})$ for $m=n+1$ and
where $e_{n-r}$ are chosen so that $\widehat{h}_{r}^{B}$ in (\ref{q1me})
satisfy the quantum Frobenius condition (\ref{fcq}). Each $\widehat{H}_{r}%
^{B}$ is then obtained dy deforming of $\widehat{h}_{r}^{B}$ in (\ref{q1me})
through an appropriate formula in (\ref{7b}) or (\ref{7c}).

\begin{theorem}
The multi-time dependent quantum canonical transformation
\begin{align}
\widehat{H}_{r}^{A}  &  =U\,\widehat{H}_{r}^{B}\,U^{\dagger}+i\hslash
U\frac{\partial U^{\dagger}}{\partial t_{r}},\ \ U=U(\lambda,t)=e^{F(\lambda
,t)}\text{, }\label{tQT}\\
F(\lambda,t)  &  =\left(  -\frac{i}{\hslash}\sum_{\gamma\in I_{\gamma}}%
\frac{1}{\gamma-m+1}d_{\gamma}(t)\sum\limits_{j=1}^{n}\lambda_{j}^{\gamma
-m+1}\right)  \text{, \ }r=1,\ldots n\nonumber
\end{align}
transforms the magnetic quantum Painlev\'{e} Hamiltonian operators
$\widehat{H}_{r}^{B}$ into the corresponding non-magnetic quantum Painlev\'{e}
Hamiltonian operators $\widehat{H}_{r}^{A}$, provided that the functions
$\zeta_{r,j}$ and $d_{\gamma}$ satisfy the same set of first order linear
PDE's as in the classical case (see Part I, Theorem 4).
\end{theorem}

This theorem is the quantum counterpart of Theorem 4 in Part I. Note that
(\ref{tQT}) is in fact covariant, just as (\ref{QT}) is.

\begin{proof}
Let us demand that the transformation (\ref{tQT}) maps the operator
$\widehat{H}_{r}^{B}$, onto the corresponding operator $\widehat{H}_{r}^{A}$
with some functions $c_{\alpha}(t)$. Consider first the Hamiltonians
$\widehat{h}_{r}^{B}$ for which $\widehat{H}_{r}^{B}=\widehat{h}_{r}^{B}$ i.e.
when $r\in\left\{  1,\ldots,\kappa_{1}\right\}  \cup\left\{  n-\kappa
_{2}+1,\ldots,n\right\}  $. In such cases, the relation (\ref{tQT}) is
satisfied provided that (cf. Part I)
\begin{equation}
\sum_{\alpha=-m}^{2n-m+2}c_{\alpha}(t)V_{r}^{(\alpha)}=\sum_{\gamma
,\gamma^{\prime}=0}^{n+1}d_{\gamma}(t)d_{\gamma^{\prime}}(t)V_{r}%
^{(\gamma+\gamma^{\prime}-m)}+\overline{b}e_{n}(t)V_{r}^{(n)}-\overline
{b}e_{n-r}(t)+S_{r}(t,\lambda)+\frac{\partial F(\lambda,t)}{\partial t_{r}%
},\quad r=1,\dotsc,n, \label{wymog}%
\end{equation}
where functions $S_{r}$ are given by%
\begin{equation}
S_{r}=U\widehat{W}_{r}U^{\dagger}-\widehat{W}_{r}=\sum_{\gamma=0}^{n+1}%
\sum_{j=1}^{n}d_{\gamma}(t)J_{r}^{j}\lambda_{j}^{\gamma-m} \label{Sr}%
\end{equation}
The condition (\ref{wymog}) is satisfied if and only if
\begin{equation}
\sum_{\alpha=-m}^{2n-m+2}c_{\alpha}(t)V_{r}^{(\alpha)}=\sum_{\gamma
,\gamma^{\prime}=0}^{n+1}d_{\gamma}(t)d_{\gamma^{\prime}}(t)V_{r}%
^{(\gamma+\gamma^{\prime}-m)}+\overline{b}e_{n}(t)V_{r}^{(n)}-\overline
{b}e_{n-r}(t) \label{w1}%
\end{equation}
and%
\begin{equation}
S_{r}(t,\lambda)+\frac{\partial F(\lambda,t)}{\partial t_{r}}=0 \label{w2}%
\end{equation}
separately. The condition (\ref{w1}) defines a map between the functions
$e_{r}(t),d_{\gamma}(t)$ and $c_{\alpha}(t)$ which reconstructs the classical
result. Since $S_{r}$ in (\ref{Sr}) is equal to its classical counterpart (see
Part I, Appendix B) and since the partial derivatives $\frac{\partial
F(\lambda,t)}{\partial t_{r}}\,\ $\ are also equal in the classical and in the
quantum case, the condition (\ref{w2}) reconstructs exactly the system of
PDE's for functions $d_{\gamma}(t)$, $\zeta_{r,j}(t)$ from Theorem 4 in Part
I. Further, in case that $r\in\left\{  \kappa_{1}+1,\ldots,n-\kappa
_{2}\right\}  $ the operator $\widehat{H}_{r}^{B}$ is the same linear
combination of appropriate operators $\widehat{h}_{s}^{B}$ as the
corresponding classical counterparts. It means that the proof of the above
theorem reduces to the proof of its classical counterpart (Theorem 4 Part I).
\end{proof}

\begin{example}
Let us choose $n=3$ and $m=1$\ and let us consider the system (7.2) in Part I
with $b_{4}=b_{2}=b_{1}=b_{0}=0$ and $b_{3}=b$ an arbitrary parameter (i.e.
the system from Example 4 from Part I). Its minimal quantization has the form
(\ref{q1me}) with $\overline{b}=0$ and explicitly reads%
\[
\widehat{h}_{r}^{B}=\widehat{E}_{r}-i\hslash J_{r}+b(t_{2}+t_{3}%
^{2})\widehat{M}_{r}^{(1)}+2bt_{3}\widehat{M}_{r}^{(2)}+b\widehat{M}_{r}%
^{(3)}\text{, \ \ }r=1,2,3
\]
where $\widehat{E}_{r}$ and $J_{r}$ are exactly as in Example \ref{4e} above,
and where $\widehat{M}_{r}^{(\gamma)}=-\frac{1}{2}i\hslash(\gamma-\frac{1}%
{2})V_{r}^{(\gamma-1)}-i\hslash(P_{r}^{(\gamma)})^{j}\partial_{j} $ with
$P_{r}^{(\gamma)}$ given by (\ref{Pr}). Then the operators $\widehat{H}%
_{1}^{B}=\widehat{h}_{1}^{B},~\widehat{H}_{2}^{B}=\widehat{h}_{2}^{B}$ and
$\widehat{H}_{3}^{B}=\widehat{h}_{3}^{B}+t_{2}\widehat{h}_{1}^{B}$ satisfy the
quantum Frobenius condition (\ref{fcq}).The unitary operator $U$ in
(\ref{tQT}) is explicitly given by $U=U(\lambda,t)=e^{F(\lambda,t)}$ with%
\begin{align*}
F(\lambda,t)  &  =-\frac{i}{\hslash}\left[  b(t_{2}+t_{3}^{2})\sum
\limits_{j=1}^{n}\lambda_{j}+bt_{3}\sum\limits_{j=1}^{n}\lambda_{j}^{2}%
+\frac{1}{3}b\sum\limits_{j=1}^{n}\lambda_{j}^{3}\right]  =\\
&  =\frac{i}{\hslash}\left[  b(t_{2}+t_{3}^{2})q_{1}+bt_{3}(2q_{2}-q_{1}%
^{2})+\frac{1}{3}b(q_{1}^{3}-3q_{1}q_{2}+3q_{3})\right]
\end{align*}
and the quantum canonical transformation (\ref{tQT}) yields the quantum
Hamiltonians $\widehat{H}_{r}^{A}$ of the form $\widehat{H}_{1}^{A}%
=\widehat{h}_{1}^{A},~\widehat{H}_{2}^{A}=\widehat{h}_{2}^{A}$ and
$\widehat{H}_{3}^{A}=\widehat{h}_{3}^{A}+t_{2}\widehat{h}_{1}^{A}$ with%
\[
\widehat{h}_{r}^{A}=\widehat{E}_{r}-i\hslash J_{r}+\sum_{\alpha=0}%
^{5}c_{\alpha}(t_{1},\dotsc,t_{n})V_{r}^{(\alpha)}%
\]
with
\begin{align*}
c_{5}  &  =\frac{1}{2}b^{2}\text{, }c_{4}=2b^{2}t_{3}\text{, }c_{3}%
=b^{2}(3t_{3}^{2}+t_{2})\text{, }c_{2}=2b^{2}(t_{2}+t_{3}^{2})t_{3}\text{,}\\
c_{1}  &  =\frac{1}{2}b^{2}(t_{2}^{2}+t_{3}^{4}+2t_{2}t_{3}^{2})+2bt_{3}%
\text{, }c_{0}=2b(t_{2}+t_{3}^{2})t_{3}+2b^{2}t_{2}t_{3}(t_{2}+t_{3}^{2}).
\end{align*}
This is exactly the non-magnetic system from Example 2 in Part I provided that
we put $a_{5}=\frac{1}{2}b^{2}$ (note however that the form of non-dynamical
parts is different above and in the mentioned Example). One can show by a
direct computation that $\widehat{H}_{r}^{A}$ do indeed satisfy the quantum
Frobenius condition.
\end{example}

\section{Conclusions}

In the series of articles (Part I-III) we have constructed multi-parameter and
multi-dimensional hierarchies of Painlev\'{e}-type systems, both classical and
quantum, from the corresponding classical respectively quantum
St\"{a}ckel-type systems. In particular, they contain the famous one-degree of
freedom Painlev\'{e} equations $P_{I}-P_{IV}$. Each hierarchy was presented in
the ordinary as well as in the magnetic regime. We also constructed the
multi-time canonical maps between both regimes (on the classical level and on
the quantum level). Also, the proposed Painlev\'{e} hierarchies $P_{I}-P_{IV}%
$, presented in Part II, can be written in the quantum version.

One of the directions of future research is to identify the obtained
hierarchies with the known hierarchies constructed by other methods, such as
appropriate reductions of soliton hierarchies. We stress that our method is
much more complete than the existing methods as in our construction, for a
given number $n$ of degrees of freedom, we obtain $n$ different
Painlev\'{e}-type systems that mutually satisfy the Frobenius integrability
condition, while the existing methods often produce only one system for a
given $n$. Another, related, direction of future research is to find a
systematic way of relating the obtained hierarchies of Painlev\'{e}-type
systems with non-autonomous and non-homogeneous soliton-type hierarchies.

\begin{acknowledgement}
MB wished to express his gratitude to Department of Science and Technology of
Link\"{o}ping University, Sweden, for their hospitality during his visits.
\end{acknowledgement}

\end{document}